\documentclass[12pt,letterpaper]{article}

\usepackage{booktabs}

\usepackage{latexsym}
\usepackage{amssymb,amsmath, bm,pgfplots,tikz,bbm}
\usepackage{graphicx}
\usepackage{algpseudocode}
\usepackage{marvosym}
\usepackage{multirow,float}
\usepackage{caption}
\usepackage{subcaption}
\usepackage{comment}
\usepackage{centernot}
\usepackage{makecell}
\usepackage{natbib}
\usepackage{color}
\usepackage[bookmarksopen=true, bookmarksnumbered=true,
pdfstartview=FitH, breaklinks=true, urlbordercolor={0 1 0}, citebordercolor={0 0 1}]{hyperref}

\usepackage{dcolumn}
\newcolumntype{.}{D{.}{.}{-1}}
\newcolumntype{d}[1]{D{.}{.}{#1}}
\newcolumntype{C}{>{$}c<{$}}

\usepackage{theorem}
\theoremstyle{plain}
\theoremheaderfont{\scshape}
\newtheorem{theorem}{Theorem}
\newtheorem{proposition}{Proposition}
\newtheorem{assumption}{Assumption}

\newtheorem{lemma}{Lemma}

\newcommand{\qed}{\hfill \ensuremath{\Box}}

\newcommand{\mc}[1]{\mathbb{#1}}
\renewcommand{\d}{\mathrm{d}}

\DeclareMathOperator\Cov{Cov}
\DeclareMathOperator\Corr{Corr}

\DeclareMathOperator*{\argmin}{argmin}

\newenvironment{proof}{\vspace{1ex}\noindent{\bf Proof}\hspace{0.5em}}
	{\hfill\qed\vspace{1ex}}
\usetikzlibrary{decorations.markings}
\usetikzlibrary{decorations.pathmorphing}
\usetikzlibrary{shapes.geometric, arrows}
\usetikzlibrary{arrows,decorations.pathmorphing,backgrounds,positioning,fit,matrix}
\usetikzlibrary{shapes,decorations,arrows,calc,arrows.meta,fit,positioning}
\tikzset{auto,node distance =1 cm and 1 cm,semithick,
	state/.style ={circle, draw, minimum width = 0.7 cm},
	point/.style = {circle, draw, inner sep=0.04cm,fill,node contents={}},
	bidirected/.style={Latex-Latex,dashed},
	el/.style = {inner sep=2pt, align=left, sloped}
}
\usetikzlibrary{positioning}
\usetikzlibrary{fadings}
\usetikzlibrary{intersections}
\usepackage{kantlipsum}
\allowdisplaybreaks

\usepackage{rotating}

\usepackage{arydshln}
\usepackage{threeparttable}

\usepackage[compact]{titlesec}

\newcommand{\blind}{0}

\usepackage{algorithm}
\usepackage{algpseudocode}
\usepackage{pifont}
\usepackage{xcolor}

\begin{document}

\newcommand\ud{\mathrm{d}}
\newcommand\dist{\buildrel\rm d\over\sim}
\newcommand\ind{\stackrel{\rm indep.}{\sim}}
\newcommand\iid{\stackrel{\rm i.i.d.}{\sim}}
\newcommand\logit{{\rm logit}}
\renewcommand\r{\right}
\renewcommand\l{\left}
\newcommand\cO{\mathcal{O}}
\newcommand\cY{\mathcal{Y}}
\newcommand\cZ{\mathcal{Z}}
\newcommand\E{\mathbb{E}}
\newcommand\cL{\mathcal{L}}
\newcommand\V{\mathbb{V}}
\newcommand\cA{\mathcal{A}}
\newcommand\cB{\mathcal{B}}
\newcommand\cD{\mathcal{D}}
\newcommand\cE{\mathcal{E}}
\newcommand\cM{\mathcal{M}}
\newcommand\cU{\mathcal{U}}
\newcommand\cN{\mathcal{N}}
\newcommand\cT{\mathcal{T}}
\newcommand\cX{\mathcal{X}}
\newcommand\bA{\bm{A}}
\newcommand\bH{\bm{H}}
\newcommand\bB{\bm{B}}
\newcommand\bP{\bm{P}}
\newcommand\bQ{\bm{Q}}
\newcommand\bU{\bm{U}}
\newcommand\bD{\bm{D}}
\newcommand\bS{\bm{S}}
\newcommand\bx{\bm{x}}
\newcommand\bX{\bm{X}}
\newcommand\bV{\bm{V}}
\newcommand\bW{\bm{W}}
\newcommand\bM{\bm{M}}
\newcommand\bZ{\bm{Z}}
\newcommand\bY{\bm{Y}}
\newcommand\bT{\bm{T}}
\newcommand\bt{\bm{t}}
\newcommand\bbeta{\bm{\beta}}
\newcommand\bpi{\bm{\pi}}
\newcommand\bdelta{\bm{\delta}}
\newcommand\bgamma{\bm{\gamma}}
\newcommand\balpha{\bm{\alpha}}
\newcommand\bone{\mathbf{1}}
\newcommand\bzero{\mathbf{0}}
\newcommand\tomega{\tilde\omega}
\newcommand{\argmax}{\operatornamewithlimits{argmax}}

\newcommand{\R}{\textsf{R}}

\newcommand\spacingset[1]{\renewcommand{\baselinestretch}%
{#1}\small\normalsize}

\spacingset{1}

\newcommand{\tit}{\bf Statistical Inference for Heterogeneous
  Treatment Effects Discovered by Generic Machine Learning in
  Randomized Experiments}


\if1\blind
\title{\tit}
\fi

\if0\blind

{\title{\tit\thanks{We thank the Sloan foundation (Economics Program;
      2020--13946) for partial support. The proposed methodology is
      implemented through an open-source \R\ package, {\sf evalITR},
      which is freely available for download at the Comprehensive R
      Archive Network (CRAN;
      \url{https://CRAN.R-project.org/package=evalITR}).  We also
      thank Kevin Du, Lucas Janson, and Yash Nair for their feedback
      on an earlier version of this paper. }}

  \author{
  Kosuke Imai\thanks{Professor, Department of Government and
      Department of Statistics, Harvard University, Cambridge, MA
      02138. Phone: 617--384--6778, Email:
      \href{mailto:Imai@Harvard.Edu}{Imai@Harvard.Edu}, URL:
      \href{https://imai.fas.harvard.edu}{https://imai.fas.harvard.edu}} \hspace{.75in}
Michael Lingzhi Li\thanks{Operation Research
	Center, Massachusetts Institute of Technology, Cambridge, MA
	02139. \href{mailto:mlli@mit.edu}{mlli@mit.edu}}
}

  \date{\today
  }

  \fi
\maketitle

\pdfbookmark[1]{Title Page}{Title Page}

\thispagestyle{empty}
\setcounter{page}{0}

\begin{abstract}

  Researchers are increasingly turning to machine learning (ML)
  algorithms to investigate causal heterogeneity in randomized
  experiments.  Despite their promise, ML algorithms may fail to
  accurately ascertain heterogeneous treatment effects under practical
  settings with many covariates and small sample size.  In addition,
  the quantification of estimation uncertainty remains a challenge.
  We develop a general approach to statistical inference for
  heterogeneous treatment effects discovered by a generic ML
  algorithm.  We apply the Neyman's repeated sampling framework to a
  common setting, in which researchers use an ML algorithm to estimate
  the conditional average treatment effect and then divide the sample
  into several groups based on the magnitude of the estimated effects.
  We show how to estimate the average treatment effect within each of
  these groups, and construct a valid confidence interval.  In
  addition, we develop nonparametric tests of treatment effect
  homogeneity across groups, and rank-consistency of within-group
  average treatment effects.  The validity of our methodology does not
  rely on the properties of ML algorithms because it is solely based
  on the randomization of treatment assignment and random sampling of
  units.  Finally, we generalize our methodology to the cross-fitting
  procedure by accounting for the additional uncertainty induced by
  the random splitting of data.

  \bigskip
  \noindent {\bf Key Words:} causal inference, causal heterogeneity,
  conditional average treatment effect, cross-fitting, multiple
  testing, randomization inference, sample splitting

\end{abstract}

\clearpage
\spacingset{1.83}
\section{Introduction}
\label{sec:intro}

A growing number of researchers are turning to machine learning (ML)
algorithms to uncover causal heterogeneity in randomized
experiments. ML algorithms are appealing because in many applications
the structure of heterogeneous treatment effects is unknown.  Despite
their promise, however, relatively little theoretical properties have
been established for many of these algorithms.  In addition, the
choice of tuning parameter values remains to be often difficult and
consequential in practice.  As a result, ML algorithms may fail to
ascertain heterogeneous treatment effects under common settings with
many covariates and small sample size.  Furthermore, one major
challenge is the quantification of statistical uncertainty when
estimating heterogeneous treatment effects using ML algorithms.

In this paper, we develop a general approach to statistical inference
for heterogeneous treatment effects estimated through the application
of a generic ML algorithm to experimental data.  We apply the
\citet{neym:23}'s repeated sampling framework to a common setting, in
which researchers use ML algorithms to estimate the conditional
average treatment effect (CATE) given pre-treatment covariates and
then divide the sample into several groups based on the magnitude of
these estimated effects.  We show how to obtain a consistent estimate
of the average treatment effect within each of these groups --- the
sorted group average treatment effect (GATES; \cite{cher:etal:19}) ---
and construct an asymptotically valid confidence interval.

We also propose two nonparametric tests of treatment effect
heterogeneity that are of interest to applied researchers.  First, we
test whether there exists any treatment effect heterogeneity across
groups.  Second, we develop a statistical test of the rank-consistency
of GATES.  If an ML algorithm produces a reasonable scoring rule, the
rank ordering of GATES based on their magnitude should be monotonic.
To accommodate the use of various ML algorithms, we make no assumption
about their properties.  Specifically, ML algorithms do not have to be
consistent or unbiased.  This is possible because the validity of our
confidence intervals and nonparametric tests solely depends on the
randomization of treatment assignment and random sampling of units.
Thus, our approach imposes only a minimal set of assumptions on the
underlying data generating process.

We first consider the settings, in which an external data set is used
to estimate the CATE.  For example, researchers may apply an ML
algorithm to an observational dataset.  Alternatively, an experimental
dataset may be split into the training and validation data sets where
an ML algorithm is first applied to the training data to estimate the
CATE, and the validation data is then used to estimate the GATES.
Here, we treat the estimated CATE function as fixed and do not account
for the uncertainty that arises from its estimation.

To incorporate this additional source of uncertainty, we further
generalize our methodology to the cross-fitting procedure, which
randomly splits the data into multiple folds.  Each fold is used as
the validation data to estimate the GATES while the remaining folds
serve as the corresponding training data to estimate the CATE.  After
repeating this for each fold, we aggregate the GATES estimates to the
entire sample.  Unlike the sample-splitting case where we condition on
the split, we account for additional uncertainty induced by the
randomness of its cross-fitting procedure.  This directly addresses
the fact that when the sample size is small the GATES estimate may
vary considerably due to the random splitting of data.

\paragraph{Related Literature.}

The proposed methodology builds on the existing literature about
statistical inference for heterogeneous treatment effects.  In an
early work, \citet{crum:etal:08} propose nonparametric tests of
treatment effect heterogeneity.  The authors rely on the consistency
of sieve methods under the assumption that heterogeneous treatment
effects are a smooth function of covariates.  In contrast, our
methodology does not require the consistent estimation of the CATE by
ML algorithms.  Moreover, while \citeauthor{crum:etal:08} assume the
continuous differentiability of the CATE, we only require its
continuity.

\citet{ding:fell:mira:16} propose an alternative approach based on
Fisher's randomization test.  Similar to our proposed methodology,
this test neither requires modeling assumptions nor imposes
restrictive assumptions on data generating process.  In fact, their
test yields conservative $p$-values without asymptotic approximation
whereas other approaches including ours are only valid in large
samples.  The authors, however, test restrictive sharp null
hypotheses.  For example, \cite{ding:fell:mira:16} consider a null
hypothesis that the individual treatment effect is constant within
each group and the effect only varies across groups.  In contrast, we
focus on the null hypotheses about average treatment effects that may
vary within and across groups under the Neyman's repeated sampling
framework.  While our tests are valid only asymptotically, our
simulation studies show that they perform reasonably well in small
samples.  In addition, \citet{ding:fell:mira:19} use the Neyman's
repeated sampling framework to explore treatment effect heterogeneity
like we do, but rely entirely on the linear regression and does not
allow for the use of more flexible ML algorithms.

More recently, \cite{cher:etal:19} study the settings that are
identical to the ones considered in this paper.  Similar to our
methodology, the authors do not impose strong assumptions on the
properties of ML algorithms that are used to estimate the CATE.
However, to incorporate the additional uncertainty of the
cross-fitting procedure, \cite{cher:etal:19} propose to repeat the
procedure many times and aggregate the resulting $p$-values.  We avoid
such a computationally intensive procedure and instead use the
Neyman's repeated sampling framework to conduct valid statistical
inference under cross-fitting.  In simulation studies reported
elsewhere \citep{imai:li:23a}, we show that our confidence intervals
are less conservative than those proposed by \cite{cher:etal:19} in
finite samples.

Other researchers also have considered GATES and related quantities.
For example, \cite{yadlowsky2021evaluating} establish the asymptotic
properties for a related general class of metrics that summarize the
effect of treatment prioritization rules.  In addition to the
different focus, the authors assume that a treatment prioritization
rule of interest is fixed and do not consider the uncertainty that
arises from its estimation. \cite{dwivedi2020stable} also estimate the
GATES to explore treatment effect heterogeneity and develop
calibration methods.  However, they do not derive the asymptotic
distribution of GATES and hence stop short of providing formal
statistical methods.

Finally, \citet{imai2019experimental} show how to evaluate an
individualized treatment rule derived from the application of a
generic ML algorithm in general settings including the one based on
cross-fitting.  We build on this work and derive the asymptotic
properties of the GATES estimator.  \citet{imai:li:23b} further
extends the methodology proposed in this paper and develop uniform
asymptotic confidence bands.  This allows researchers to choose, with
a statistical guarantee, a group of individuals who are predicted to
benefit from or be harmed by the treatment, using the estimated CATE
based on a generic ML algorithm.  They do not, however, consider the
estimation uncertainty of the CATE.

\section{The Proposed Methodology}
\label{sec:sample_splitting}

We start by developing our methodology in a setting where the
conditional average treatment effect (CATE) function is estimated
using a separate data set, but is considered fixed when estimating the
sorted group average treatment effect (GATES) and conducting
statistical tests. For instance, the estimated CATE might come from an
external, possibly observational, dataset. An alternative is
\textit{sample splitting}, where the sample is divided randomly into
training and evaluation sets. The training data is used for CATE
estimation via a machine learning algorithm, and the evaluation data
for GATES estimation. In this section, we do not account for the
uncertainty in estimating the CATE. In
Section~\ref{sec:cross_fitting}, we extend our methodology to
cross-fitting, incorporating this estimation uncertainty.

\subsection{Setup}
\label{subsec:randexp}

Suppose that we have an independently and identically distributed
(i.i.d.) sample of $n$ units from a super-population $\mathcal{P}$.
Let $T_i$ represent the treatment assignment indicator variable, which
is equal to $1$ if unit $i$ is assigned to the treatment condition and
is equal to $0$ otherwise, i.e., $T_i \in \cT = \{0,1\}$. For each
unit, we observe the outcome variable $Y_i \in \cY$ and a vector of
pre-treatment covariates, $\bX_i \in \cX$, where $\cY$ and $\cX$
represent the support of the outcome variable and that of the
pre-treatment covariates, respectively.

We require the standard causal inference assumptions of consistency
and no interference between units, denoting the potential outcome for
unit $i$ under the treatment condition $T_i = t$ as $Y_i(t)$ for
$t=0,1$ \citep[e.g.,][]{neym:23,holl:86,rubi:90}.  The observed
outcome is given by $Y_i = Y_i(T_i)$.  For notational simplicity, we
assume that the treatment assignment is completely randomized with
exactly $n_1$ units assigned to the treatment condition though the
extensions to other experimental designs and unconfounded
observational designs are possible.  We formally state these
assumptions below.
\begin{assumption}[No Interference between Units] \label{asm:SUTVA}
	\spacingset{1} The potential outcomes for unit $i$ do not depend on
	the treatment status of other units.  That is, for all
	$t_1, t_2,\ldots,t_n \in \{0, 1\}$, we have,
	$Y_i(T_1 = t_1, T_2 = t_2, \ldots, T_n = t_n) = Y_i(T_i = t_i).$
\end{assumption}
\begin{assumption}[Random Sampling of Units] \label{asm:randomsample}
	\spacingset{1} Each of $n$ units, represented by a three-tuple
	consisting of two potential outcomes and pre-treatment covariates,
	is assumed to be independently sampled from a super-population
	$\mathcal{P}$, i.e.,
	$$(Y_i(1), Y_i(0), \bX_i) \ \iid \ \mathcal{P}$$
\end{assumption}
\begin{assumption}[Complete Randomization] \label{asm:comrand}
  \spacingset{1} For any $\bt \in \{0,1\}^n$ such that $\sum_{i=1}^n \bt_i=n_1$, the
  treatment assignment probability is given by,
	$$\Pr(\bT = \bt \mid \{Y_{i^\prime}(1), Y_{i^\prime}(0),
        \bX_{i^\prime}\}_{i^\prime=1}^n) \ = \ \frac{1}{{n_1\choose n}}.$$
\end{assumption}

Suppose that a researcher applies an ML algorithm to a training
dataset and estimate the CATE.  As noted earlier, this training
dataset can be obtained through the sample splitting or it may be an
external dataset.  The CATE is defined as,
\begin{equation*}
  \tau(\bx) \ = \ \E(Y_i(1)-Y_i(0) \mid \bX_i = \bx), \label{eq:CATE}
\end{equation*}
for any $\bx \in \cX$.  The ML algorithm produces the
following scoring rule,
\begin{equation}
  s: \cX \longrightarrow \mathcal{S} \subset \mathbb{R} \label{eq:scoring}
\end{equation}
where a greater score indicates a higher priority to receive the
treatment.  Without loss of generality, we assume that the scoring
rule is bijective, i.e., $s(\bx) \ne s(\bx^\prime)$ for any
$\bx, \bx^\prime \in \cX$ with $\bx \ne \bx^\prime$.  Note that one
can always redefine $\cX$ to satisfy this assumption.

As noted earlier, we assume almost nothing about the properties of
this scoring rule derived by the ML algorithm.  In particular, the
scoring rule does not have to be a consistent estimate of CATE.  In
fact, the scoring rule need not even be an estimate of CATE so long as
it satisfies the definition given in Equation~\eqref{eq:scoring}.

\subsection{Estimation and Inference}
\label{subsec:GATE}

Given the setup introduced above, we first consider the estimation and
inference for the sorted group average treatment effect (GATES), which
is a common quantity of interest in applied research and is studied by
\cite{cher:etal:19}.  The idea is that researchers sort units into a
total of $K$ groups based on the quantile of the scoring rule, and
then estimate the average treatment effect within each group.  For
simplicity, we assume that the number of treated and control units,
i.e., $n_1$ and $n_0$, are multiples of $K$.  The formal definition of
GATES is given by,
\begin{equation}
  \tau_k \ = \ \mc{E}(Y_i(1)-Y_i(0)\mid c_{k-1}(s)< s(\bX_i) \leq c_{k}(s))\label{eq:gCATE}
\end{equation}
for $k=1,2,\ldots,K$ where $c_k$ represents the cutoff between the
$(k-1)$th and $k$th groups and is defined as,
\begin{equation*}
  c_k(s) \ = \ \inf\{c \in \mc{R} \mid \Pr(s(\bX_i)\leq c)\geq k/K\},
\end{equation*}
for $k=1,2,\ldots,K$, with $c_0 = -\infty$. Equivalently, GATES can be
seen as a special case of the rank-weighted average treatment effect
(RATE) with $\alpha(u)=\mathbf{1}\{\frac{k-1}{K}< u\leq \frac{k}{K}\}$
\citep{yadlowsky2021evaluating}.

Thus, units that belong to the $K$th group, for example, represent
those who are likely to have the greatest treatment effect according
to the ML algorithm whereas those in the first group are likely to
have the least treatment effect.  However, we do not assume that the
GATES is monotonic, i.e., $\tau_k \le \tau_{k^\prime}$ for all
$k < k^\prime$.  This is important because we want to impose as little
restriction on the underlying scoring rule as possible.  Indeed, if
the scoring rule is not a good estimate of CATE, such an assumption
may be violated.  To address this problem, we later develop a
statistical test of this monotonicity assumption.

We consider the following estimator of GATES using the experimental
data,
\begin{equation}
\hat \tau_{k} \ = \ \frac{K}{n_1} \sum_{i=1}^n Y_i T_i \hat{f}_k(\bX_i) -
\frac{K}{n_0} \sum_{i=1}^n Y_i(1-T_i)\hat{f}_k(\bX_i),  \label{eq:gCATEest}
\end{equation}
for $k=1,2,\ldots,K$ where
$\hat{f}_k(\bX_i) = \mathbf{1}\{s(\bX_i) > \hat{c}_{k-1}(s)\} -
\mathbf{1}\{s(\bX_i) > \hat{c}_{k}(s)\}$, and
$\hat{c}_{k}(s) \ = \ \inf \{c \in \mathbb{R}: \sum_{i=1}^n
\mathbf{1}\{s(\bX_i)\leq c\} \geq  nk/K\}$ is the estimated cutoff.  First,
we derive the bias bound and exact variance of the GATES estimator.
\begin{theorem} {\sc (Bias Bound and Exact Variance of the GATES
    Estimator)} \label{thm:GATEest} \spacingset{1.25} Under
  Assumptions~\ref{asm:SUTVA}--\ref{asm:comrand}, the bias of the
  proposed estimator of GATES given in Equation~\eqref{eq:gCATEest}
  can be bounded as follows,
	\begin{eqnarray*}
		& & \mc{P}(|\E\{\hat \tau_{k} -\tau_{k}\mid \hat{c}_{k}(s),\hat{c}_{k-1}(s) \} |\geq
		\epsilon)  \\
		& \leq &   1-B\l(\frac{k}{K}+\gamma_{k}(\epsilon), \frac{nk}{K}
		, n- \frac{nk}{K} +1\r)+B\l(\frac{k}{K}-\gamma_{k}(\epsilon),
		\frac{nk}{K}
		, n-\frac{nk}{K} +1\r)\\& & -B\l(\frac{k-1}{K}+\gamma_{k-1}(\epsilon), \frac{n(k-1)}{K}
		, n-\ \frac{n(k-1)}{K} +1\r) \\ & & +B\l(\frac{k-1}{K}-\gamma_{k-1}(\epsilon),
		\frac{n(k-1)}{K}
		,  n-\frac{n(k-1)}{K} +1\r),
	\end{eqnarray*}
	for any given constant $\epsilon > 0$ where
	$B(\epsilon, \alpha, \beta)$ is the incomplete beta function
	(if $\alpha \leq 0$ and $\beta > 0$, we set
	$B(\epsilon,\alpha, \beta):=H(\epsilon)$ for all $\epsilon$
	where $H(\epsilon)$ is the Heaviside step function), and
	\begin{equation*}
		\gamma_{k}(\epsilon)\ = \ \frac{\epsilon}{K\max_{c  \in
				[ c_{k}(s) -\epsilon,\   c_{k}(s) +\epsilon]} \E(Y_i(1)-Y_i(0)\mid s(\bX_i)=c)}.
	\end{equation*}
	The variance of the estimator is given by,
	\begin{equation*}
		\V(\hat \tau_{k}) \  = \
		K^2\left(\frac{\E(S_{k1}^2)}{n_1} +
		\frac{\E(S_{k0}^2)}{n_0}\right) +
                \frac{(n-K)\kappa_{k11}}{n-1} -
                \kappa_{k1}^2,
	\end{equation*}
	where
        $S_{kt}^2 = \sum_{i=1}^n (Y_{ki}(t) -
        \overline{Y_{k}(t)})^2/(n-1)$,
        $\kappa_{kt} = \E(Y_i(1)-Y_i(0)\mid \hat{f}_k(\bX_i)=t)$, and
        $\kappa_{ktt} = \E[(Y_i(1)-Y_i(0))(Y_j(1)-Y_j(0)) \mid
        \hat{f}_k(\bX_i)=\hat{f}_k(\bX_j)=t]$ for $i\ne j$ with
        $Y_{ki}(t) = \hat{f}_k(\bX_i)Y_i(t)$, and
        $\overline{Y_{k}(t)} = \sum_{i=1}^n Y_{ki}(t)/n$, for
        $t= 0,1$.
\end{theorem}
Proof is given in Supplementary Appendix~\ref{app:GATEest}.

When compared to the standard variance estimator, there are additional
two terms.  These terms result from the fact that the cutoff points
are estimated, yielding a cross-unit correlation in terms of
$\hat{f}_k(\bX_i)Y_i(t)$. Since exactly $n/K$ data points are taken to
have $\hat{f}_k(\bX_i)=1$, the value of this function is generally
negatively correlated across units, i.e.,
$\Corr(\hat{f}_k(\bX_i),\hat{f}_k(\bX_j))<0$.

The variance can be consistently estimated by replacing each unknown
parameter with its sample analogue:
\begin{eqnarray*}
\widehat{\E(S_{kt}^2)} & = &
\frac{1}{n_t-1} \sum_{i=1}^n \mathbf{1}\{T_i = t\} (Y_{ki} -
\overline{Y_{kt}})^2, \\ 
\hat\kappa_{kt} & = & \frac{\sum_{i=1}^n \mathbf{1}\{\hat{f}_k(\bX_i) =
	t\} T_i  Y_i}{\sum_{i=1}^n  \mathbf{1}\{\hat{f}_k(\bX_i)=
	t\} T_i } - \frac{\sum_{i=1}^n   \mathbf{1}\{\hat{f}_k(\bX_i) =  t\} (1-T_i)  Y_i}{\sum_{i=1}^n
	\mathbf{1}\{\hat{f}_k(\bX_i) =  t\}
	(1-T_i)}, \\
\hat{\kappa}_{ktt} & = & \frac{[\sum_{i=1}^n
                             \mathbf{1}\{\hat{f}_k(\bX_i) =
                             t\}T_iY_i]^2-\sum_{i=1}^n
                             \mathbf{1}\{\hat{f}_k(\bX_i) =
                             t\}T_iY_i^2}{[\sum_{i=1}^n
                             \mathbf{1}\{\hat{f}_k(\bX_i) =
                             t\}T_i]^2-\sum_{i=1}^n
                             \mathbf{1}\{\hat{f}_k(\bX_i) =  t\}T_i} \\
                       & & + \frac{[\sum_{i=1}^n
                           \mathbf{1}\{\hat{f}_k(\bX_i) =
                           t\}(1-T_i)Y_i]^2-\sum_{i=1}^n
                           \mathbf{1}\{\hat{f}_k(\bX_i) =
                           t\}(1-T_i)Y_i^2}{[\sum_{i=1}^n
                           \mathbf{1}\{\hat{f}_k(\bX_i) =
                           t\}(1-T_i)]^2-\sum_{i=1}^n
                           \mathbf{1}\{\hat{f}_k(\bX_i) =  t\}(1-T_i)} \\
                       & & - 2\frac{[\sum_{i=1}^n   \mathbf{1}\{\hat{f}_k(\bX_i) =  t\}(1-T_i)Y_i][\sum_{i=1}^n   \mathbf{1}\{\hat{f}_k(\bX_i) =  t\}T_iY_i]}{[\sum_{i=1}^n   \mathbf{1}\{\hat{f}_k(\bX_i) =  t\}(1-T_i)][\sum_{i=1}^n   \mathbf{1}\{\hat{f}_k(\bX_i) =  t\}T_i]}.
\end{eqnarray*}
for $t=0,1$ where $Y_{ki} = \hat{f}_k(\bX_i)Y_i$ and
$\overline{Y}_{kt} = \sum_{i=1}^n \mathbf{1}\{T_i = t\} Y_{ki}/n_t$.
The expression of $\hat{\kappa}_{ktt}$ above enables the calculation
in $O(n)$ rather than $O(n^2)$ time.  The details of the derivation is
given in Appendix~\ref{app:kappa}.

We can further derive the asymptotic sampling distribution of the GATE
estimator by requiring the following continuity assumption and moment
conditions:
\begin{assumption}[Continuity of CATE at the
  Thresholds] \label{asm:continuity} \spacingset{1} Let
  $F(c)=\Pr(s(\bX_i)\leq c)$ represent the cumulative distribution
  function of the scoring rule and define its pseudo-inverse
  $F^{-1}(p)=\inf \{c: F(c) \ge p\}$ for $p\in [0,1]$.  The CATE
  function $\E(Y_i(1)-Y_i(0) \mid s(\bX_i) = F^{-1}(p))$ is assumed to
  be left-continuous with bounded variation on any interval
  $(\theta, 1-\theta)$ with $\theta>0$, and continuous in $p$ at
  $p=1/K,\ldots, (K-1)/K$.
\end{assumption}
\begin{assumption}[Moment Conditions] \label{asm:moments}
	\spacingset{1.25} For each $t=0,1$, we have
	\begin{enumerate}
		\item $\V(Y_i(t))>0$;
		\item $\E(Y_i(t)^3) < \infty$.
	\end{enumerate}
\end{assumption}
Assumption~\ref{asm:continuity} is similar to the assumption commonly
used in the literature that the CATE is continuous in the covariates
$\bX_i$ \citep[e.g.,][]{kunz:etal:18,wage:athe:18}, but we only
require continuity at the thresholds, $1/K,\cdots, (K-1)/K$ and
bounded variation everywhere else.  We will show in
Proposition~\ref{prop:continuity_tight} below that
Assumption~\ref{asm:continuity} is among the weakest assumptions
necessary for our asymptotic results. In particular, this assumption
requires that the scoring rule cannot be discontinuous at the
thresholds unless the CATE is constant in the scoring rule, i.e.
$\E(Y_i(1)-Y_i(0) \mid s(\bX_i) = F^{-1}(p))=\E(Y_i(1)-Y_i(0))$ for
all $p$.

We now present the asymptotic sampling distribution of GATES
estimator.
\begin{theorem}[Asymptotic Sampling Distribution of GATES Estimator]
	\label{thm:GATEsamp} \spacingset{1.25}
	Under Assumptions~\ref{asm:SUTVA}--\ref{asm:moments},
	we have,
	\begin{equation*}
		\frac{\hat{\tau}_k-\tau_k}{\sqrt{\V(\hat\tau_k)}} \
		\stackrel{d}{\longrightarrow} \ N(0,1)
	\end{equation*}
	for $k=1,\ldots,K$ where $\V(\hat\tau_k)$ is given in
        Theorem~\ref{thm:GATEest}.
\end{theorem}
Proof is given in Supplementary Appendix~\ref{app:GATEsamp}.  We
emphasize that Theorem~\ref{thm:GATEsamp} does not impose a strong
assumption about the properties of the ML algorithm used to generate
the scoring rule $s$.

In fact, the continuity of the CATE at the thresholds
(Assumption~\ref{asm:continuity}) is among the weakest assumptions
that can ensure the validity of Theorem~\ref{thm:GATEsamp}. To see
this, consider an alternative assumption that there exists a threshold
at which CATE is bounded but discontinuous, slightly relaxing
Assumption~\ref{asm:continuity}. The following proposition shows that
this assumption is not sufficient for Theorem~\ref{thm:GATEsamp}.
\begin{proposition}[Insufficiency of Bounded
	Variation]\label{prop:continuity_tight} \spacingset{1.25}
	Suppose Assumptions~\ref{asm:SUTVA}--\ref{asm:comrand} and
        \ref{asm:moments} hold.  Further assume that the there exists
        a threshold $k/K$, such that
        $\E(Y_i(1)-Y_i(0) \mid s(\bX_i) = F^{-1}(p))$, is
        discontinuous (but bounded) at $p=k/K$. Then, there exist a
        scoring rule $s$ and a population $\mathcal{P}$ such that as
        $n\to \infty$ with $0 < n_1/n < 1$ staying constant, we have,
	\begin{equation*}
		\E\l(	\frac{\hat{\tau}_k-\tau_k}{\sqrt{\V(\hat\tau_k)}}\r) \ \centernot\longrightarrow \  0.
	\end{equation*}
\end{proposition}
Proof is given in Supplementary Appendix
\ref{app:continuity_tight}. Proposition~\ref{prop:continuity_tight}
shows that if the CATE is mildly discontinuous at a threshold, then we
cannot sufficiently control the bias in estimating the boundary
points, $c_k(s)$.  Under this scenario, the bias decays at the rate of
$n^{-1/2}$, which is not fast enough for the application of the
central limit theorem.

\subsection{Nonparametric Test of Treatment Effect Heterogeneity}
\label{subsec:heterotest}

In many applications, heterogeneous treatment effects are imprecisely
estimated.  Researchers may wish to know whether the treatment effect
heterogeneity discovered by ML algorithms represents signal rather
than noise.  In addition, checking the statistical significance of
each GATES suffers from multiple testing problems.  To address these
challenges, we develop a nonparametric test of treatment effect
heterogeneity.  In particular, we consider the following null
hypothesis that all GATEs are equal to one another,
\begin{equation}
H_0: \ \tau_1 = \tau_2 = \cdots = \tau_K. \label{eq:H0}
\end{equation}
This null hypothesis is equivalent to $\tau_k = \tau$ for any $k$
where $\tau = \E(Y_i(1)-Y_i(0))$ represents the overall average
treatment effect (ATE).  Thus, we consider the following test
statistic,
\begin{equation*}
  \bm{\hat{\tau}} \ = \ (\hat\tau_1-\hat\tau,\cdots,\hat\tau_K-\hat\tau)^\top,
\end{equation*}
where
\begin{equation*}
  \hat \tau \ = \ \frac{1}{n_1} \sum_{i=1}^n Y_iT_i - \frac{1}{n_0} \sum_{i=1}^n Y_i(1-T_i).
\end{equation*}

To derive the asymptotic reference distribution of this test
statistic,

\cite{imai2019experimental} derive the bias bound and the exact
variance of this PAPE estimator.  Leveraging those results,
the following theorem shows that we can utilize a $\chi^2$
distribution as an asymptotic approximation to the reference
distribution when testing treatment effect heterogeneity.
\begin{theorem}[Nonparametric Test of Treatment Effect Heterogeneity]
  \label{thm:heterotest} \spacingset{1.25} Suppose
  Assumptions~\ref{asm:SUTVA}--\ref{asm:moments} hold. Under $H_0$
  defined in Equation~\eqref{eq:H0} and against the alternative $H_1: \mc{R}^K
  \setminus H_0$, as $n \to \infty$ with $0<n_1/n<1$
  stays constant, we have,
	\begin{equation*}
\bm{\hat{\tau}}^\top \bm{\Sigma}^{-1}\bm{\hat{\tau}} \
\stackrel{d}{\longrightarrow} \ \chi^2_K
	\end{equation*}
	where the entries of the covariance matrix $\bm{\Sigma}$ are
        defined as follows,
	\begin{align*}
\bm{\Sigma}_{kk} \  = & \
                      K^2\left[\frac{\E(S_{k1}^{*2})}{n_1} +
\frac{\E(S_{k0}^{*2})}{n_0} \right.\\ & \hspace{.5in} +\left. \frac{1}{K^3}\left\{(K-2)\left(\frac{n-K}{n-1}\kappa_{kk11}-\kappa_{k1}^2\right)-\frac{2n(K-1)}{(n-1)}\kappa_{kk01}+2\kappa_{k1}\kappa_{k0}\right\} \right],\\
\bm{\Sigma}_{kk^\prime} \  = & \
K^2\left\{\frac{\E(S_{kk^\prime 1}^{*2})}{n_1} +
\frac{\E(S_{kk^\prime 0}^{*2})}{n_0}  \right\}
                               +\frac{1}{K}\biggl\{(K-2)\left(\kappa_{kk'11}-\kappa_{k1}\kappa_{k'1}\right)
                               \\ & \hspace{.5in} \left. -\frac{Kn-n-1}{n-1}\left(\kappa_{kk'10}+\kappa_{kk'01}\right)+\kappa_{k1}\kappa_{k'0}+\kappa_{k0}\kappa_{k'1}\right\},
\end{align*}
for $k,k^\prime \in \{1,\cdots K\}$ and $k\neq k^\prime$ where
$S_{kt}^{*2} = \sum_{i=1}^n (Y_{ki}^\ast(t) -
\overline{Y_{k}^\ast(t)})^2/(n-1)$,
$S_{kk^\prime t}^{*2} = \sum_{i=1}^n (Y_{ki}^\ast(t) -
\overline{Y_{k}^\ast(t)})(Y_{k^\prime i}^\ast(t) -
\overline{Y_{k^\prime}^\ast(t)})/(n-1)$,
$\kappa_{kt} = \E(Y_i(1)-Y_i(0)\mid \hat{f}_k(\bX_i)=t)$, and
$\kappa_{kk'ts} = \E[(Y_i(1)-Y_i(0))(Y_j(1)-Y_j(0)) \mid
\hat{f}_k(\bX_i)=t,\hat{f}_{k'}(\bX_i)=s]$ for $i \ne j$ with
$Y_{ki}^\ast(t) = (\hat{f}_k(\bX_i)- 1/K)Y_i(t)$, and
$\overline{Y_{k}^\ast(t)} = \sum_{i=1}^n Y_{ki}^\ast(t)/n$, for
$t= 0,1$.
\end{theorem}
Proof is given in Supplementary Appendix~\ref{app:heterotest}. Similar to
Theorem~\ref{thm:GATEest}, there is an additional third term in the
variance beyond the two standard terms, induced by the fact that
$\hat{f}_k(\bX_i)$ is negatively correlated across units. In practice,
we replace the entries of $\bm{\Sigma}$ with their sample analogues,
which result in a consistent estimator $\bm{\widehat\Sigma}$.  By
Slutsky's Lemma, the asymptotic distribution is not affected by this
substitution.

\subsection{Nonparametric Test of Rank-Consistent Treatment Effect Heterogeneity}
\label{subsec:consistest}

To evaluate the quality of the scoring rule produced by an ML
algorithm, we can test whether or not the rank of estimated GATES is
consistent with that of the true GATES.  The relevant null hypothesis
is given by,
\begin{equation}
  H_0^\ast: \tau_1\leq \tau_2\leq\cdots \leq \tau_{K}. \label{eq:H0ast}
\end{equation}
Unlike the null hypothesis for treatment effect heterogeneity given in
Equation~\eqref{eq:H0}, this is a composite null hypothesis.

To characterize the sampling distribution under this null hypothesis
$H_0^\ast$, we consider the following optimization problem,
\begin{align*}
  \bm{\mu}^\ast (\bm{x}) \ = \ \;&\argmin_{\bm{\mu}} \|\bm{\mu}-\bm{x}\|_2^2 \quad
                            \text{subject to } \mu_1\leq \mu_2 \leq
                             \cdots \leq \mu_{K},
\end{align*}
where $\bm{\mu}=(\mu_1,\mu_2,\ldots,\mu_K)^\top$ and
$\bm{x} \in \mc{R}^K$.  If $\bm{x} \sim N(0,\bm{\Sigma})$, the following
test statistic has a mixture of appropriately weighted $\chi^2$
distribution with $K$ degrees of freedom, called chi-bar-squared
distribution \citep{shapiro1988towards},
\begin{equation*}
  (\bm{x}-\bm{\mu}^\ast(\bm{x}))^\top \bm{\Sigma}^{-1} (\bm{x}-\bm{\mu}^\ast(\bm{x})) \
  \sim \ \bar\chi_K^2.
\end{equation*}
Using this fact, the next theorem derives a nonparametric test of
rank-consistent treatment effect heterogeneity that is asymptotically
uniformly most powerful.
\begin{theorem} {\sc (Nonparametric Test of Rank-Consistent Treatment Effect
  Heterogeneity)} \label{thm:consistest} \spacingset{1.25}
  Suppose that Assumptions~\ref{asm:SUTVA}--\ref{asm:moments} hold.
  Then, as $n\to \infty$ and $0<n_1/n <1$ stays constant, an
  asymptotically uniformly most powerful test of size $\alpha$ for the
  null hypothesis $H_0^\ast$ defined in Equation~\eqref{eq:H0ast} against
  the alternative $H_1^\ast: \mc{R}^K\setminus H_0^\ast$ has the
  following critical region,
	\begin{equation*}
	\{\bm{\hat\tau} \in \mc{R}^K \mid
        \left(\bm{\hat\tau}-\bm{\mu}_0(\bm{\hat\tau})\right)^\top\bm{\Sigma}^{-1}\left(\bm{\hat\tau}-\bm{\mu}_0(\bm{\hat\tau})\right)>C_\alpha\},
	\end{equation*}
	for some constant $C_\alpha$ that only depends on $\alpha$.
        The expression of $\bm{\Sigma}$ is given in
        Theorem~\ref{thm:heterotest}. Under $H_0^\ast$ and as
        $n\to \infty$, we have,
	\begin{equation*}
          \left(\bm{\hat\tau}-\bm{\mu}^\ast(\bm{\hat\tau})\right)^\top \bm{\Sigma}^{-1}\left(\bm{\hat\tau}-\bm{\mu}^\ast(\bm{\hat\tau})\right) \stackrel{d}{\longrightarrow} \bar\chi^2_K.
	\end{equation*}
\end{theorem}
Proof is given in Supplementary Appendix~\ref{app:consistest}. In
practice, we use Monte Carlo simulations to approximately compute the
critical values.

While our test is the asymptotically most powerful test of its type,
it is likely to be conservative as we control the critical value based
on the worst-case scenario among all the distributions consistent with
the null hypothesis. In the literature on statistical tests of moment
inequalities, scholars have developed subsampling and moment selection
techniques that can improve their statistical power \citep[see
e.g.,][]{andrews2009validity,andrews2010inference,canay2010inference,chernozhukov2019inference}.
\cite{canay2023user} provides an up-to-date review.

\section{Generalization to Cross-Fitting}
\label{sec:cross_fitting}

In this section, we generalize our methodology to
\emph{cross-fitting}, in which researchers use the same experimental
data to first generate the scoring rule using an ML algorithm and then
estimate the GATES based on the resulting scoring rule.  In comparison
with {\it sample splitting} discussed in
Section~\ref{sec:sample_splitting} where they are done on separate
samples, cross-fitting could potentially be much more efficient.  The
key challenge, however, is the incorporation of additional uncertainty
due to the random splitting of the data.  We show how to overcome this
under the Neyman's repeated sampling framework.

\subsection{Estimation and Inference}
\label{subsec:GATEcv}

Under cross-fitting, we randomly divide the experimental data into
$L \ge 2$ folds of equal size $m=n/L$ where for the sake of simplicity
we assume $n$ is a multiple of $L$, and each fold contains $m_1$
treated units with $m_0$ control units, i.e, $m=m_0+m_1$.  We maintain
Assumptions~\ref{asm:SUTVA}--\ref{asm:comrand} introduced in
Section~\ref{subsec:randexp}.  Then, for each $\ell = 1,2,\ldots, L$,
we use the $\ell$th fold as a validation dataset
$\cZ_\ell=\{\bX_i^{(\ell)}, T_i^{(\ell)}, Y_i^{(\ell)}\}_{i=1}^{m}$ to
conduct statistical tests and estimate the GATES.  We use the remaining
folds,
$\cZ_{-\ell}=\{\bX_i^{(-\ell)}, T_i^{(-\ell)},
Y_i^{(-\ell)}\}_{i=1}^{n-m}$, as the training dataset to estimate the
scoring rule with an ML algorithm.

Suppose that we define a generic ML algorithm as a deterministic map
from the space of training data
$\boldsymbol{\mathcal{Z}}_{\text{train}}$ to the space of scoring
rules $\mathcal{S}$:
\begin{equation*}
	F: \boldsymbol{\mathcal{Z}}_{\text{train}} \to \mathcal{S}.
\end{equation*}
Then, for a given training data set $\cZ_{\text{train}}$ of size $n-m$, the estimated scoring rule is
given by,
\begin{equation}
  \hat{s}_{\cZ_{\text{train}}^{n-m}} \ = \ F(\cZ_{\text{train}}^{n-m}). \label{eq:MLtr}
\end{equation}

We now generalize the definition of the GATES to the cross-fitting case,
\begin{equation}
	\tau_{k}(F, n-m) \ = \ \mc{E}[\mc{E}\{Y_i(1)-Y_i(0)\mid c_{k-1}(\hat{s}_{\cZ_{\text{train}}^{n-m}})\leq \hat{s}_{\cZ_{\text{train}}^{n-m}}(\bX_i) \leq c_{k}(\hat{s}_{\cZ_{\text{train}}^{n-m}})\}]\label{eq:gCATEtr},
\end{equation}
where the inner expectation is taken over the distribution of
$\{\bX_i, Y_i(0), Y_i(1)\}$ among the units who belong to the $k$th
group, and the outer expectation is taken over all possible training
sets of size $n-m$ from $\cZ_{\text{train}}^{n-m}$ the population
$\mathcal{P}$.

This generalized GATES is not a function of fixed scoring rule. Rather,
it is a function of ML algorithm $F$ itself (as well as the sample
size of training data, $n-m$). Intuitively, it represents the average
of GATES based on all observations that score between
$(k-1)/K \times 100$th percentile and $k/K \times 100$th percentile
under the ML algorithm $F$ across all possible training datasets of
size $n-m$. Alternatively, the cross-fitted GATE can be seen as a
weighted average of GATEs that are specific to scoring rules where
weights are determined by the training data and the particular ML
algorithm.

We describe estimation and inference for $\tau_k(F,n-m)$.  For each
fold $\ell$, we first estimate a scoring rule $s$ by applying an ML
algorithm $F$ to the training data $\cZ_{-\ell}$,
\begin{equation}
	\hat{s}_{\ell} \ = \ F(\cZ_{-\ell}). \label{eq:MLcv}
\end{equation}
We then estimate the GATES based on the validation data $\cZ_\ell$,
using the following estimator that is analogous to the one defined in
Equation~\eqref{eq:gCATEest},
\begin{eqnarray*}
  & & \hat \tau_{k}^\ell(F, n-m) \\
  & = & K\left[\frac{1}{m_1} \sum_{i=1}^m Y_i^{(\ell)} T_i^{(\ell)} \hat{f}^\ell_k(\bX_i^{(\ell)}) +
	\frac{1}{m_0} \sum_{i=1}^m
                                         Y_i^{(\ell)}(1-T_i^{(\ell)})\left\{1-\hat{f}^\ell_k(\bX_i^{(\ell)})\right\} - \frac{1}{m_0} \sum_{i=1}^m Y_i^{(\ell)}(1-T_i^{(\ell)})\right],
\end{eqnarray*}
where
$\hat{f}^\ell_k(\bX_i^{(\ell)}) =
\mathbf{1}\{\hat{s}_\ell(\bX_i^{(\ell)}) \geq
\hat{c}^\ell_{k-1}(\hat{s}_\ell)\} -
\mathbf{1}\{\hat{s}_\ell(\bX_i^{(\ell)}) \geq
\hat{c}^\ell_{k}(\hat{s}_\ell)\}$, and
$\hat{c}^\ell_{k}(\hat{s}_\ell) = \inf \{c \in \mathbb{R}:
\sum_{i=1}^m \mathbf{1}\{\hat{s}_\ell(\bX_i^{(\ell)})>c\} \le mk/K\}$
represents the estimated cutoff in the $\ell$th subsample.  Repeating
this for each fold and averaging the results gives us the final GATES
estimator,
\begin{equation}
	\hat \tau_k (F, n-m)\ = \  \frac{1}{L} \sum_{\ell=1}^L \hat \tau_{k}^\ell \label{eq:gCATEestcv}
\end{equation}
for $k = 1,2,\ldots,K$. Algorithm \ref{alg:cross_validation}
summarizes this estimation procedure.

\begin{algorithm}[!t]
	\spacingset{1}
	\hspace*{\algorithmicindent} \textbf{Input}: Data $\cZ=\{\bX_i, T_i,
	Y_i\}_{i=1}^n$, Machine learning algorithm $F$, Estimator
	$\hat\tau_k$, Number of folds $L$ \\
	\hspace*{\algorithmicindent} \textbf{Output}: Estimated GATES $\{\hat\tau_{k}(F, n-m)\}_{k=1}^K$
	\begin{algorithmic}[1]
		\State Split the data $\cZ$ into $L$ random subsets of equal size $\{\cZ_1,\cdots,\cZ_L\}$
		\State Set $m \gets n/L$ and $\ell\gets 1$
		\While{$\ell\leq L$}
		\State $\cZ_{-\ell}=\{\cZ_1, \cdots,
                \cZ_{\ell-1},\cZ_{\ell+1}, \cdots, \cZ_{L}\}$
                \Comment{Create the training dataset}
		\State $\hat{s}_{-\ell} \ = \ F(\cZ_{-\ell})$
                \Comment{Estimate the scoring rule $s$
			by applying $F$ to $\cZ_{-\ell}$}
		\State $ \hat{\tau}_k^\ell =
                \hat{\tau}_{k}(\cZ_\ell)$ for each $k=1,2,\ldots,K$
                \Comment{Calculate the GATES estimator using $\cZ_\ell$}
		\State $\ell \gets \ell + 1$
		\EndWhile
		\State \textbf{return} $\hat{\tau}_{k}(F, n-m) =
                \frac{1}{L} \sum_{\ell=1}^L \hat{\tau}_k^\ell$ for
                each $k=1,2,\ldots,K$
	\end{algorithmic}
	\caption{Estimation of the Sorted Group Average Treatment Effects
          (GATES) under Cross-fitting}
	\label{alg:cross_validation}
\end{algorithm}

We extend our bias and variance results under sample splitting
(Theorem~\ref{thm:GATEest}) to the cross-fitting case by incorporating
the additional randomness induced by the cross-fitting procedure.
\begin{theorem} {\sc (Bias Bound and Exact Variance of the GATES
    Estimator under Cross-fitting)} \label{thm:GATEestcv}
  \spacingset{1.25} Under
  Assumptions~\ref{asm:SUTVA}--\ref{asm:comrand}, the bias of the
  proposed GATES estimator given in Equation~\eqref{eq:gCATEestcv} can
  be bounded as follows,
	\begin{eqnarray*}
		& &
                    \E\l[\mc{P}\l(\l|\E\{\hat
                    \tau_{k}(F,n-m) -\tau_{k}(F,n-m)\mid
                    \hat{c}_{k}(\hat{s}_{\cZ_{\text{train}}^{n-m}}),\hat{c}_{k-1}(\hat{s}_{\cZ_{\text{train}}^{n-m}})
                    \} \r | \geq
		\epsilon\ \Bigl | \ \cZ_{\text{train}}^{n-m}\r)\r]  \\
			& \leq &   1-B\l(\frac{k}{K}+\gamma_{k}(\epsilon), \frac{nk}{K}
		, n- \frac{nk}{K} +1\r)+B\l(\frac{k}{K}-\gamma_{k}(\epsilon),
		\frac{nk}{K}
		, n-\frac{nk}{K} +1\r)\\& & -B\l(\frac{k-1}{K}+\gamma_{k-1}(\epsilon), \frac{n(k-1)}{K}
		, n-\ \frac{n(k-1)}{K} +1\r) \\ & & +B\l(\frac{k-1}{K}-\gamma_{k-1}(\epsilon),
		\frac{n(k-1)}{K}
		,  n-\frac{n(k-1)}{K} +1\r),
	\end{eqnarray*}
	for any given constant $\epsilon > 0$
        where $B(\epsilon, \alpha, \beta)$ is the incomplete beta
        function (if $\alpha \leq 0$ and $\beta > 0$, we set
        $B(\epsilon,\alpha, \beta):=H(\epsilon)$ for all $\epsilon$
        where $H(\epsilon)$ is the Heaviside step function), and
	\begin{equation*}
		\gamma_{k}(\epsilon)\ = \ \frac{\epsilon}{K\E\{\max_{c  \in
				[ c_{k}(\hat{s}_{\cZ_\text{train}^{n-m}}(\bX_i)) -\epsilon,\   c_{k}(\hat{s}_{\cZ_\text{train}^{n-m}}(\bX_i)) +\epsilon]} \E(Y_i(1)-Y_i(0)\mid \hat{s}_{\cZ_\text{train}^{n-m}}(\bX_i)=c)\}}.
	\end{equation*}
	The variance of the estimator is given by,
	\begin{align*}
          & \V(\hat \tau_{k}(F,n-m)) \\
          =
          & \ K^2\left(\frac{\E(S_{Fk1}^2)}{m_1} +
            \frac{\E(S_{Fk0}^2)}{m_0}\right)
            +\frac{(n-K)\E_\ell(\kappa^\ell_{k11})}{n-1}
            - \E_\ell[(\kappa^\ell_{k1})^2]+\V\l(\kappa_{k1}^\ell\r)
            -\frac{L-1}{L} \E(S^2_{Fk}),
	\end{align*}
	where
        $S_{Fkt}^2 = \sum_{i=1}^m (Y^\ell_{ki}(t) -
        \overline{Y^\ell_{k}(t)})^2/(m-1)$,
        $S_{Fk}^2 = \sum_{\ell=1}^L (\hat{\tau}^{(\ell)}_k -
        \hat{\tau}_{k}(F,n-m))^2/(L-1)$,
        $\kappa_{kt}^\ell = \E(Y_i(1)-Y_i(0)\mid
        \hat{f}^\ell_k(\bX_i)=t)$, and
        $\kappa_{ktt}^\ell = \E[(Y_i(1)-Y_i(0))(Y_j(1)-Y_j(0))\mid
        \hat{f}^\ell_k(\bX_i)=\hat{f}^\ell_k(\bX_j)=t]$ for $i \ne j$
        with
        $Y^\ell_{ki}(t) =
        \hat{f}^\ell_k(\bX_i^{(\ell)})Y_i^{(\ell)}(t)$, and
        $\overline{Y^\ell_{k}(t)} = \sum_{i=1}^m Y_{ki}^\ell(t)/n$,
        for $t= 0,1$.
\end{theorem}
Proof is given in Supplementary Appendix~\ref{app:GATEestcv}. When
compared to Theorem~\ref{thm:GATEest}, although the bias bound is of a
similar form, the variance expression implies two additional terms.
The first additional term, $\V\l(\kappa_{k1}^\ell\r)$, accounts for
the variation across training data sets.  The second negative term,
$-(L-1)\E(S^2_{Fk})/L$, represents the efficiency gain of the
cross-fitting procedure.  As expected, when $L=1$, the expression
reduces to the sample splitting case (see Theorem~\ref{thm:GATEest}).

The estimation of $\E(S_{Fkt}^2)$, $\E\{(\kappa_{kt}^\ell)^2\}$,
$\E\{(\kappa_{ktt}^\ell)\}$ and $\V(\kappa_{kt}^\ell)$ is
straightforward and based on their sample analogues:
\begin{align*}
\widehat{\E(S_{Fkt}^2)}  \ &= \ \frac{1}{(m-1)L} \sum_{\ell=1}^L \sum_{i=1}^m \mathbf{1}\{T_i^{(\ell)} = t\} (Y_{ki}^\ell-
\overline{Y_{kt}^\ell})^2,  \nonumber\\
\widehat{\E\l\{(\kappa_{kt}^\ell)^2\r\}}  \ &= \ \frac{1}{L} \sum_{\ell=1}^L \left(\hat{\kappa}^\ell_{kt}\right)^2,\quad
\widehat{\V(\kappa_{kt}^\ell)}  \ = \ \frac{1}{L-1} \sum_{\ell=1}^L (\hat{\kappa}^\ell_{kt}-\overline{\hat{\kappa}^\ell_{kt}})^2,  \label{eq:kappacvest}
\end{align*}
where $Y_{ki}^\ell= \hat{f}_k^\ell(\bX_i)Y_i^{(\ell)}$,
$\overline{Y_{kt}^\ell} = \sum_{i=1}^m \mathbf{1}\{T_i = t\}
Y_{ki}^{(\ell)}/m$,
$\overline{\hat{\kappa}^\ell_{kt}}= \sum_{\ell=1}^L
\hat{\kappa}^\ell_{kt}/L$ and
\begin{align*}
  \hat{\kappa}^\ell_{kt} = & \frac{\sum_{i=1}^m
	\mathbf{1}\{\hat{f}_k^\ell(\bX_i^{(\ell)}) = t\}  T_i^{(\ell)} Y_i^{(\ell)}}{\sum_{i=1}^m
	\mathbf{1}\{\hat{f}_k^\ell(\bX_i^{(\ell)}) = t\}  T_i^{(\ell)} }  - \frac{\sum_{i=1}^m
	\mathbf{1}\{\hat{f}_k^\ell(\bX_i^{(\ell)}) = t\}  (1-T_i^{(\ell)}) Y_i^{(\ell)}}{\sum_{i=1}^m \mathbf{1}\{\hat{f}_k^\ell(\bX_i^{(\ell)}) = t\}  (1-T_i^{(\ell)})},\\
\hat{\kappa}^\ell_{ktt} = &   \frac{[\sum_{i=1}^m   \mathbf{1}\{\hat{f}^\ell_k(\bX_i^{(\ell)}) =  t\}T_i^{(\ell)}Y_i^{(\ell)}]^2-\sum_{i=1}^m   \mathbf{1}\{\hat{f}^\ell_k(\bX_i^{(\ell)}) =  t\}T_i^{(\ell)}(Y_i^{(\ell)})^2}{[\sum_{i=1}^m   \mathbf{1}\{\hat{f}^\ell_k(\bX_i^{(\ell)}) =  t\}T_i^{(\ell)}]^2-\sum_{i=1}^m   \mathbf{1}\{\hat{f}^\ell_k(\bX_i^{(\ell)}) =  t\}T_i} \nonumber\\ &- \frac{[\sum_{i=1}^m   \mathbf{1}\{\hat{f}^\ell_k(\bX_i^{(\ell)}) =  t\}(1-T_i^{(\ell)})Y_i^{(\ell)}]^2-\sum_{i=1}^m   \mathbf{1}\{\hat{f}^\ell_k(\bX_i^{(\ell)}) =  t\}(1-T_i^{(\ell)})(Y^{(\ell)}_i)^2}{[\sum_{i=1}^m   \mathbf{1}\{\hat{f}^\ell_k(\bX_i^{(\ell)}) =  t\}(1-T_i^{(\ell)})]^2-\sum_{i=1}^m   \mathbf{1}\{\hat{f}^\ell_k(\bX_i^{(\ell)}) =  t\}(1-T_i^{(\ell)})}, \nonumber
\end{align*}

In contrast, the estimation of $\E(S^2_{Fk})$ requires care.  In
particular, although it is tempting to estimate $\E(S^2_{Fk})$ using a
realization of $S^2_{Fk}$, this estimate is highly variable especially
when $L$ is small.  As a result, it often yields a negative overall
variance estimate.  We address this problem by applying Lemma 1 from \cite{nadeau2000inference} to
$\hat \tau_{k}(F,n-m)$, which gives,
\begin{equation*}
\V(\hat \tau_{k}(F,n-m)) \geq \ \E(S^2_{Fk}).
\end{equation*}
Since Theorem~\ref{thm:GATEestcv} implies:
\begin{equation*}
	\V(\hat \tau_{k}(F,n-m)) \leq \ K^2\left(\frac{\E(S_{Fk1}^2)}{m_1} +
	\frac{\E(S_{Fk0}^2)}{m_0}\right)+\frac{(n-K)\E_\ell[\kappa^\ell_{k11}]}{n-1} -	\E_\ell[(\kappa^\ell_{k1})^2]+\V\l(\kappa_{k1}^\ell\r),
\end{equation*}
this inequality suggests the following estimator of
$\E(S^2_{Fk})$,
\begin{equation}
\widehat{\E(S^2_{Fk})} \ = \ \min\l(S^2_{Fk}, K^2\left(\frac{\widehat{\E(S_{Fk1}^2)}}{m_1} +
\frac{\widehat{\E(S_{Fk0}^2)}}{m_0}\right)+\frac{(n-K)\E_\ell[\kappa^\ell_{k11}]}{n-1} -	\E_\ell[(\kappa^\ell_{k1})^2]+\widehat{\V\l(\kappa_{k1}^\ell\r)}\r). \label{eq:conserv_esti}
\end{equation}
Although this yields a conservative estimate of
$\V(\hat \tau_{k}(F,n-m))$ in finite samples, the bias appears to be
relatively small in practice (see Section~\ref{sec:synthetic}). In
Appendix \ref{app:consistent}, we show that the estimator is
consistent as $L$ goes to infinity and sufficiently large $m$.

To establish the asymptotic sampling distribution of our cross-fitting
GATES estimator, we first extend our CATE continuity condition
(Assumption~\ref{asm:continuity}) by assuming that each CATE given a
training data set is continuous and the average CATE (over all
possible training data sets) is bounded.
\begin{assumption}[Continuity of CATE at the Thresholds under
  Cross-Fitting] \label{asm:continuity_cv} \spacingset{1} Let
  $F_{\cZ^{n-m}_{\text{train}}}(c)=\Pr(\hat{s}_{\cZ^{n-m}_{\text{train}}}(\bX_i)\leq
  c)$ represent the cumulative distribution function of the scoring
  rule under training set $\cZ^{n-m}_{\text{train}}$ and define its
  pseudo-inverse as
  $F_{\cZ^{n-m}_{\text{train}}}^{-1}(p)=\inf \{c:
  F_{\cZ^{n-m}_{\text{train}}}(c) \ge p\}$ for $p\in [0,1]$.  Then,
  for all but asymptotically measure-zero set of possible training
  sets $\cZ^{n-m}_{\text{train}}$ of size $n-m$, the CATE function
  $\tau_{\cZ^{n-m}_{\text{train}}}(p) = \E(Y_i(1)-Y_i(0) \mid
  \hat{s}_{\cZ^{n-m}_{\text{train}}}(\bX_i) =
  F_{\cZ^{n-m}_{\text{train}}}^{-1}(p))$ is left-continuous with
  bounded variation on any interval $(\epsilon, 1-\epsilon)$ with
  $0< \epsilon <1/2$, and continuous in $p$ at
  $p=1/K,\ldots, (K-1)/K$. Furthermore, we assume
  $\lim_{n \to \infty} \E_{\cZ^{n-m}_{\text{train}}}[\max_{p \in
    [0,1]} \tau_{\cZ^{n-m}_{\text{train}}}(p)]<\infty$.
\end{assumption}

In addition, we require the ML algorithm $F$ to be stable.
\begin{assumption}[ML Algorithm Stability]\label{asm:stability} \spacingset{1.25}
  Let $\cZ_{\text{train}}^n$ be a training dataset of size $n$ and
  $\hat{s}_{\cZ_{\text{train}}^n}=F(\cZ_{\text{train}}^n)$ represent
  the estimated scoring rule that results from the application of an
  ML algorithm $F$ to the training dataset. Then, as $m \to \infty$
  (with $L$ fixed), for any $a,b$ with $a < b$:
	\begin{equation*}
	\|\E[Y_i(1)-Y_i(0) \mid a\leq
		\hat{s}_{\cZ_{\text{train}}^{n}}(\bX_i) \leq b]\|_{2}=o\left(m^{-1}\right).
	\end{equation*}
\end{assumption}
The expectation is taken over the distribution of
$\{\bX_i, Y_i(0), Y_i(1)\}$ among those units in the population
$\mathcal{P}$ who belong to the group defined by the conditioning set.
The outer norm is computed across the random sampling of training data
set of size $n$ from the same population.
Assumption~\ref{asm:stability} implies that as the size of training
data approaches infinity, $L_2$ norm of the resulting scoring rule
$\hat{s}_{\cZ_{\text{train}}^{n}}$ stabilizes sufficiently quickly at
a rate faster than $O(m^{-1})$. The required rate is consistent with
the asymptotic conditions needed for other related cross-validation
settings \citep[e.g.,][]{austern2020asymptotics}. Importantly, we do
not assume that the ML algorithm converges to the true CATE.

Finally, the next theorem established the asymptotic distribution of
GATES estimator under cross-fitting.
\begin{theorem} {\sc (Asymptotic Sampling Distribution of GATES
    Estimator under Cross-Fitting)}
  \label{thm:GATEsampcv} \spacingset{1.25} Suppose $L$ is fixed.
  Then, under Assumptions
  ~\ref{asm:SUTVA}--\ref{asm:comrand},~\ref{asm:moments}--\ref{asm:stability},
  we have, as $m$ goes to infinity,
	\begin{equation*}
		\frac{\hat{\tau}_{k}(F,n-m) -\tau_{k}(F,n-m)}{\sqrt{\V(\hat{\tau}_{k}(F,n-m))}} \
		\stackrel{d}{\longrightarrow} \ N(0,1)
	\end{equation*}
	where the expression of $\V(\hat\tau_k(F,n-m))$ is given in
        Theorem~\ref{thm:GATEestcv}.
\end{theorem}
Proof is given in Supplementary Appendix~\ref{app:GATEsampcv}, and is similar to the
proof of Theorem~\ref{thm:GATEsamp}.

\subsection{Nonparametric Tests of Treatment Effect Heterogeneity}

We now extend the nonparametric tests of treatment effect
heterogeneity and its rank-consistency introduced in
Sections~\ref{subsec:heterotest}~and~\ref{subsec:consistest} to the
cross-fitting setting.  Similar to \citet{cher:etal:19}, we account
for the additional uncertainty due to random splitting.
Unlike their method, however, the proposed tests do not require a
computationally intensive resampling procedure.

Our first null hypothesis of interest is that the GATES are all equal
to the ATE,
\begin{equation}
	H_{F0} : \ \tau_1(F,n-m) \ = \ \tau_2(F,n-m) \ = \ \cdots \
                     = \ \tau_K(F,n-m). \label{eq:H0cross}
\end{equation}
This null hypothesis depends on the ML algorithm $F$ whereas the null
hypothesis given in Equation~\eqref{eq:H0} depends on the (fixed)
scoring rule.  

The following theorem generalizes the result of
Theorem~\ref{thm:heterotest} to cross-fitting.
\begin{theorem}{\sc (Nonparametric Test of Treatment Effect
    Heterogeneity Under Cross-fitting)}
  \spacingset{1.25} \label{thm:heterotestcv} Suppose $L$ is fixed. Then, under
  Assumptions~\ref{asm:SUTVA}--\ref{asm:comrand},~\ref{asm:moments}--\ref{asm:stability},
  and the null hypothesis $H_{F0}$ defined in
  Equation~\eqref{eq:H0cross} and against the alternative
  $H_{F1}: \mc{R}^K \setminus H_{F0}$, as $m\to \infty$, and
  $0<m_1/m <1$ stays constant, we have,
	\begin{equation*}
		\bm{\hat{\tau}}_F^\top \bm{\Sigma}^{-1}\bm{\hat{\tau}}_F
                \ \stackrel{d}{\longrightarrow} \ \chi^2_K
	\end{equation*}
	where
        $\bm{\hat{\tau}}_F=(\hat\tau_{1}(F,n-m)-\hat\tau,\cdots,\hat\tau_{K}(F,n-m)-\hat\tau)$,
        and $\bm{\Sigma}$ is defined as for $k,k^\prime \in \{1,\cdots K\}$:
	\begin{align*}
		\bm{\Sigma}_{kk} \  = & \
K^2\left(\frac{\E(S_{Fk1}^{*2})}{m_1} +
\frac{\E(S_{Fk0}^{*2})}{m_0}\right) -\frac{L-1}{L}\E(S^{2}_{Fk})+\V\l(\kappa_{k1}^\ell\r)\\
\ & + \frac{1}{K}\E_\ell\left\{(K-2)\left(\frac{n-K}{n-1}\kappa^\ell_{kk11}-(\kappa^\ell_{k1})^2\right)-\frac{2n(K-1)}{(n-1)}\kappa^\ell_{kk01}+2\kappa^\ell_{k1}\kappa^\ell_{k0}\right\}\\		
		\bm{\Sigma}_{kk^\prime} \  = & \
		K^2\left(\frac{\E(S_{Fkk^\prime 1}^{*2})}{m_1} +
		\frac{\E(S_{Fkk^\prime 0}^{*2})}{m_0}\right) -\frac{L-1}{L}\E(S^{2}_{Fkk'})+\Cov\l(\kappa_{k1}^\ell,\kappa_{k^\prime1}^\ell\r)\\
		\ &+ \
                    \frac{1}{K}\E_\ell\left\{(K-2)\left(\kappa^\ell_{kk'11}-\kappa^\ell_{k1}\kappa^\ell_{k'1}\right)-\frac{Kn-n-1}{n-1}\left(\kappa^\ell_{kk'10}+\kappa^\ell_{kk'01}\right)+\kappa^\ell_{k1}\kappa^\ell_{k'0}+\kappa^\ell_{k0}\kappa^\ell_{k'1}\right\}
	\end{align*}
	where
        $S_{Fkt}^{*2} = \sum_{i=1}^m (Y^{*\ell}_{ki}(t) -
        \overline{Y^{*\ell}_{k}(t)})^2/(m-1)$,
        $S_{Fkk^\prime t}^{*2} = \sum_{i=1}^m (Y^{*\ell}_{ki}(t) -
        \overline{Y^{*\ell}_{k}(t)})(Y^\ell_{k^\prime i}(t) -
        \overline{Y^{*\ell}_{k^\prime}(t)})/(m-1)$,
        $S_{Fkk'}^{2} = \sum_{\ell=1}^L (\hat{\tau}^\ell_k(F,n-m) -
        \hat{\tau}_{k}(F,n-m)) (\hat{\tau}^\ell_{k'}(F,n-m) -
        \hat{\tau}_{k'}(F,n-m))/(L-1)$,
        $\kappa_{kt}^\ell = \E(Y_i(1)-Y_i(0)\mid
        \hat{f}^\ell_k(\bX_i)=t)$ and $\kappa^\ell_{kk'ts} = \E[(Y_i(1)-Y_i(0))(Y_j(1)-Y_j(0)) \mid \hat{f}^\ell_k(\bX_i)=t,\hat{f}^\ell_{k'}(\bX_i)=s]$with
        $Y^{*\ell}_{ki}(t) =
        (\hat{f}^\ell_k(\bX_i^{(\ell)})-1/K)Y_i^{(\ell)}(t)$, and
        $\overline{Y^{*\ell}_{k}(t)} = \sum_{i=1}^m
        Y_{ki}^{*\ell}(t)/m$, for $t= 0,1$.
\end{theorem}
Proof is given in Supplementary Appendix~\ref{app:heterotestcv}.
Compared to Theorem~\ref{thm:heterotest}, the only difference appears
in the expression of the covariance matrix $\bm{\Sigma}$ due to the
efficiency gains of the cross-validation procedure. Similar to
Theorem~\ref{thm:GATEestcv}, the estimation of $\E(S^2_{Fkk'})$ for
$k=k'$ requires care, and we utilize the consistent estimator as
identified in Equation \eqref{eq:conserv_esti}.  If the resulting
covariance matrix estimate is not positive definite, we find the
nearest positive definite matrix in the $L_2$ norm by utilizing the
algorithm of \cite{higham2002computing}.

Finally, we extend the nonparametric test of rank-consistent treatment
effect heterogeneity (Theorem~\ref{thm:consistest}) to cross-fitting.
The null hypothesis is given by,
\begin{align}
	H_{F0}^\ast &: \tau_1(F,n-m)\leq \tau_2(F,n-m)\leq\cdots \leq \tau_K(F,n-m).\label{eq:H0cosis-cross}
\end{align}
Now, we present the result.
\begin{theorem} {\sc (Nonparametric Test of Rank-Consistent Treatment
    Effect Heterogeneity Under
    Cross-Fitting)} \label{thm:consistestcv} \spacingset{1.25} Suppose $L$ is fixed. Then, under
Assumptions~\ref{asm:SUTVA}--\ref{asm:comrand},~\ref{asm:moments}--\ref{asm:stability},
  as $m\to \infty$ and $0<m_1/m <1$ stays constant, the uniformly most
  powerful test of size $\alpha$ for the null hypothesis $H_{F0}^\ast$
  defined in Equation~\eqref{eq:H0cosis-cross} against the alternative
  $H_{F1}^\ast : \mc{R}^K \setminus H_{F0}^\ast$ has the following
  critical region,
  	\begin{equation*}
		\{\bm{\hat\tau}_F \in \mc{R}^K \mid
                \left(\bm{\hat\tau}_F -\bm{\mu}_0(\bm{\hat\tau}_F )\right)^\top
                \bm{\Sigma}^{-1}\left(\bm{\hat\tau}_F -\bm{\mu}_0(\bm{\hat\tau}_F )\right)>C_\alpha\},
	\end{equation*}
	for some constant $C_\alpha$ that only depends on
        $\alpha$. Furthermore, under $H_{F0}$ and as $n\to \infty$, we have,
	\begin{equation*}
		\left(\bm{\hat\tau}_F -\bm{\mu}_0(\bm{\hat\tau}_F )\right)^\top\bm{\Sigma}^{-1}\left(\bm{\hat\tau}_F -\bm{\mu}_0(\bm{\hat\tau}_F )\right) \stackrel{d}{\longrightarrow} \bar\chi^2_K,
	\end{equation*}
	where $\bm{\hat{\tau}}_F$ and $\bm{\Sigma}$ are defined in
        Theorem~\ref{thm:heterotestcv}.
\end{theorem}
Proof directly follows from the fact by
Theorem~\ref{thm:heterotestcv}, $\bm{\Sigma}^{-1/2}\bm{\hat\tau}_F$ is
asymptotically normally distributed with variance-covariance matrix
$\bm{I}$, which is an identity matrix of size $K \times K$. Then,
following the same steps as those in Supplementary
Appendix~\ref{app:consistest} immediately establishes the result.

\section{A Simulation Study}
\label{sec:synthetic}

We undertake a simulation study to examine the finite sample
performance of the proposed methodology.  We consider both
sample-splitting and cross-fitting cases.  For the estimation of
GATES, we evaluate the bias and variance of the proposed estimators as
well as the coverage of their confidence intervals.  For hypothesis
tests, we examine the actual power and size of the proposed tests.  We
show that the proposed methodology performs well even when the sample
size is as small as $100$.

\subsection{The Setup}

We utilize the data generating process from the 2016 Atlantic Causal
Inference Conference (ACIC) Competition.  We briefly describe its
simulation setting here and refer interested readers to
\citet{dori:etal:19} for additional details.  The focus of this
competition was the inference of average treatment effect in
observational studies. There are a total of 58 pre-treatment
covariates $\bX$, including 3 categorical, 5 binary, 27 count data,
and 13 continuous variables.  The data were taken from a real-world
study with the sample size $n=4802$.

In our simulation, we assume that the empirical distribution of these
covariates represent the population $\mathcal{P}$ and obtain each
simulation sample via bootstrap.  We consider small and moderate
sample sizes: $n= 100$, $500$, and $2500$. For the sample-splitting
case, the models are pre-trained on the original dataset from the 2016
ACIC data challenge, and the sample size $n$ refers to the number of
testing samples. For the cross-validation case, $n$ refers to the
total dataset size, which we then conduct 5-fold cross-validation,
$L=5$.  One important change from the original competition is that
instead of utilizing a propensity model to determine $T$, we assume
that the treatment assignment is completely randomized, i.e.,
$\Pr(T_i = 1)=1/2$, and the treatment and control groups are of equal
size, i.e., $n_1=n_0=n/2$.

To generate the outcome variable, we use one of the settings from the
competition, which is based on the generalized additive model with
polynomial basis functions.  The model represents a setting, in which
there exists a substantial amount of treatment effect heterogeneity.
The formula for this outcome model is reproduced here:
{\small\begin{align*}
	\E(Y_i(t) \mid \bX_i) &= 1.60+
                                0.53\times x_{29}-3.80\times x_{29}(x_{29}-0.98)(x_{29}+0.86)
                                -0.32 \times \bone\{x_{17}>0\}\\& +
  0.21 \times \bone\{x_{42}>0\}-0.63 \times  x_{27}+4.68 \times
  \bone\{x_{27}<-0.61\}-0.39 \times (x_{27}+0.91)
  \bone\{x_{27}<-0.91\}\\&+ 0.75 \times
  \bone\{x_{30}\leq0\}-1.22 \times \bone\{x_{54}\leq0\}+0.11 \times x_{37}
  \bone\{x_{4}\leq0\}-0.71 \times \bone\{x_{17}\leq0, t=0\}\\&-1.82 \times \bone\{x_{42}\leq
  0,t=1\}+0.28 \times \bone\{x_{30}\leq
  0,t=0\}\\&+\{0.58\times x_{29}-9.42 \times
  x_{29}(x_{29}-0.67)(x_{29}+0.34)\}\times\bone\{t=1\}\\&+(0.44 \times
  x_{27}-4.87\times
  \bone\{x_{27}<-0.80\})\times\bone\{t=0\}-2.54 \times \bone\{t=0,
  x_{54}\leq 0\}.
\end{align*}}

Throughout, we set $K=5$ so that observations are sorted into five
groups based on the magnitude of the CATE.  For the case of
sample-splitting, we can directly compute the true values of GATES
using the outcome model and evaluate each quantity based on the entire
original data set. For the cross-validation case, however, the exact
calculation of GATES true values would require averaging over all
combinations of training data sets from the original data set.  Since
this is computationally prohibitive, we obtain their approximate true
values by independently sampling 10,000 training data sets.  For each
training dataset, we train an ML algorithm $F$ using 5-fold
cross-validation. Then, we use the sample mean of each estimated
causal quantity across the 10,000 simulated data sets as our
approximate truth.

We evaluate Bayesian Additive Regression Trees (BART)
\citep[see][]{chipman2010bart,hill2011bayesian,hahn2020bayesian} and
Causal Forest \citep{wage:athe:18}, and LASSO \citep{tibs:96}. The
number of trees were tuned through the 5-fold cross validation for
both algorithms. For implementation, we use {\sf R 3.6.3} with the
following packages: {\sf bartMachine} (version 1.2.6) for BART, {\sf
  grf} (version 2.0.1) for Causal Forest, and {\sf glmnet} (version
4.1-2) for LASSO.  The number of trees was tuned through 5-fold
cross-validation for BART and Causal Forest. The regularization
parameter was tuned similarly for LASSO.

\subsection{Finite-Sample Performance of the Proposed Estimators}
\label{subsec:GATEest}

\begin{table}[t!]
	\centering\setlength{\tabcolsep}{1pt}
	\small \spacingset{1}
	\begin{tabular}{l.|...|...|...}
		\hline
		& & \multicolumn{3}{c|}{$\bm{n_{\text{test}}=100}$} & \multicolumn{3}{c|}{$\bm{n_{\text{test}}=500}$} & \multicolumn{3}{c}{$\bm{n_{\text{test}}=2500}$}\\
		Estimator & \multicolumn{1}{c|}{truth}
		& \multicolumn{1}{c}{bias} & \multicolumn{1}{c}{s.d.}
		& \multicolumn{1}{c|}{coverage}
				& \multicolumn{1}{c}{bias} & \multicolumn{1}{c}{s.d.}
		& \multicolumn{1}{c|}{coverage}
		& \multicolumn{1}{c}{bias} & \multicolumn{1}{c}{s.d.}
& \multicolumn{1}{c}{coverage} \\ \hline
		\multicolumn{2}{l|}{\textbf{Causal Forest}} &  & & & & &\\
		$\hat\tau_1$ & 2.164  & 0.034 & 2.486   & 93.8\%  & 0.041  &  1.071  & 95.0\%  & 0.007  &  0.467 &  96.0\%  \\
		$\hat\tau_2$ & 4.001  & 0.011 &  2.551  & 93.7    & -0.060  & 1.183  & 94.4    &  -0.002  &  0.510 &  95.3    \\
		$\hat\tau_3$ & 4.583  & -0.018 &  2.209 & 94.0    & -0.003  &  0.956 & 96.4    &  0.020  &  0.421&  95.8    \\
		$\hat\tau_4$ & 4.931  & -0.077 &  2.500 & 94.6    & 0.001  &  1.138  & 94.3    & 0.003  &  0.517 &  95.6    \\
		$\hat\tau_5$ & 5.728  & -0.058 &  3.332 & 96.0    & -0.010  &  1.499 & 95.0    & -0.009  &  0.661  &  95.2   \\\hline
		\multicolumn{2}{l|}{\textbf{BART}} &  & & & & &\\
		$\hat\tau_1$ & 2.092  & 0.016  & 3.188 & 94.0\%  & -0.014   &  1.402 & 96.2\%   & 0.009  &  0.626 &  95.8\% \\
		$\hat\tau_2$ & 3.913  & 0.127 & 2.918 & 95.1    & 0.028  & 1.388  & 94.0    &  -0.003 &  0.618 &  95.3   \\
		$\hat\tau_3$ & 4.478  & -0.077 & 2.218 & 94.3    & -0.041  &  0.968 & 95.0    &  -0.001  &  0.425 &  95.1   \\
		$\hat\tau_4$ & 5.042  & -0.154 & 2.366 & 94.2    & 0.014   &  1.106 & 95.8    & 0.015  &  0.495 &  95.4   \\
		$\hat\tau_5$ & 5.881  & -0.019 & 2.510 & 94.7    & -0.019   &  1.104 & 94.4    & -0.000   &  0.489 &  95.0  \\\hline
		\multicolumn{2}{l|}{\textbf{LASSO}} &  & & & & &\\
$\hat\tau_1$ & 3.243  & 0.028  & 2.507 & 94.1\%  & 0.049   &  1.119 & 95.1\%   & 0.003  &  0.769 &  95.1\% \\
$\hat\tau_2$ & 3.817  & -0.012 & 1.848 & 93.6    & -0.013  & 0.834  & 94.5    &  -0.000 &  0.684 &  95.4   \\
$\hat\tau_3$ & 4.318  & -0.013 & 2.095 & 94.2    & -0.002  &  0.930 & 94.5    &  0.010  &  0.516 &  95.0   \\
$\hat\tau_4$ & 4.788  & -0.041 & 2.475 & 94.0    & -0.015   &  1.101 & 94.6    & -0.001  &  0.480 &  94.6   \\
$\hat\tau_5$ & 5.241  & -0.046 & 3.921 & 94.4    & 0.021   &  1.739 & 95.1    & 0.002   &  0.505 &  95.3  \\\hline\hline
	\end{tabular}
	\caption{The Finite Sample Performance of the GATES Estimators
          under Sample-splitting.  The table presents the estimated
          bias and standard deviation of the GATES estimators as well
          as the empirical coverage of their 95\% confidence intervals
          for Causal Forest, BART, and LASSO.  The machine learning
          algorithms are trained on the original dataset from the 2016
          ACIC data challenge.} \label{tb:simulation}
\end{table}

Table~\ref{tb:simulation} presents the results for the estimation of
GATES in the sample-splitting case. According to this simulation
setup, Causal Forest and BART appear to identify treatment effect
heterogeneity better than LASSO.  For example, for BART, the largest
and smallest GATES are 5.89 and 2.09, respectively.  In contrast, the
gap between the corresponding quantities is much smaller for the
LASSO, roughly equaling 2 points.

For each sample size, we conducted 1,000 simulation trials.  For all
three algorithms, the estimated biases of the proposed GATES
estimators are negligibly small, accounting for less than 5\% of their
estimated standard deviation in almost all cases.  The bias also
generally decreases as the sample size grows.  We also find that the
empirical coverage of the confidence intervals is close to the
theoretical 95\% value even when the sample size is as small as
$n=100$.

\begin{table}[t!]
  \centering\setlength{\tabcolsep}{1pt}
     \spacingset{1}\small
     \resizebox{1.2\textwidth}{!}{
     \begin{tabular}{l|....|....|....}
		\hline
		& \multicolumn{4}{c|}{$\bm{n=100}$} & \multicolumn{4}{c|}{$\bm{n=500}$} & \multicolumn{4}{c}{$\bm{n=2500}$}\\
		Estimator
		& \multicolumn{1}{c}{truth} & \multicolumn{1}{c}{bias} & \multicolumn{1}{c}{s.d.}
		& \multicolumn{1}{c|}{coverage}
		& \multicolumn{1}{c}{truth}  & \multicolumn{1}{c}{bias} & \multicolumn{1}{c}{s.d.}
		& \multicolumn{1}{c|}{coverage}
		& \multicolumn{1}{c}{truth}   & \multicolumn{1}{c}{bias} & \multicolumn{1}{c}{s.d.}
		& \multicolumn{1}{c}{coverage} \\ \hline
			\textbf{Causal Forest} &  & & & & & & & & & & \\
		$\hat\tau_1$ & 3.976  & -0.053  & 2.971  & 94.0\% & 2.900   & -0.007   &  1.572 & 95.6\% & 2.210   & -0.007  &  0.594  &  97.7\%\\
		$\hat\tau_2$ & 4.173  & -0.061  & 2.584  & 95.9   & 4.112   & -0.038  & 1.075  & 98.2   & 4.057   &  0.011   &  0.541 &  98.6   \\
		$\hat\tau_3$  &4.286  & -0.012  & 2.560  & 96.7   & 4.510   & -0.054   &  1.058 & 97.7   & 4.545   &  0.019   &  0.465 &  98.1   \\
		$\hat\tau_4$ & 4.400  & -0.119  & 2.865  & 97.4   & 4.799   & 0.066   &  1.149 & 97.9   & 4.951   &  -0.009   &  0.509 &  98.6   \\
		$\hat\tau_5$ & 4.569  & 0.140   & 3.447  & 94.1   & 5.086   & 0.001   &  1.620 & 96.0   & 5.643   & -0.006  &  0.620  &  98.3  \\\hline
					\textbf{LASSO} &  & & & & & & & & & & \\
		$\hat\tau_1$ & 4.191  & -0.125  & 3.196  & 97.6\% & 4.017   & -0.025   &  1.488 & 96.0\% & 3.752   & -0.004  &  0.669  &  96.0\%\\
		$\hat\tau_2$ & 4.205  & 0.036  & 2.281  & 97.5   & 4.137  & -0.069  & 1.027  & 97.9   & 4.028   &  -0.019   &  0.590 &  98.9   \\
		$\hat\tau_3$  &4.268  & -0.126  & 2.354  & 96.6   & 4.291   & -0.019   &  1.000 & 97.9   & 4.323   &  0.037   &  0.488 &  97.5   \\
		$\hat\tau_4$ & 4.334  & -0.003  & 2.536 & 96.8   & 4.430   & 0.035   &  1.174 & 96.8   & 4.571   &  0.033   &  0.642 &  97.2   \\
		$\hat\tau_5$ & 4.406  & 0.111   & 3.615  & 96.2   & 4.530   & 0.047   &  1.811 & 95.0   & 4.732   & 0.022  &  0.697  &  95.3  \\\hline\hline
	\end{tabular}
}
\caption{The Finite Sample Performance of the GATES Estimators under
  Cross-fitting.  The table presents the estimated bias and standard
  deviation of the proposed GATES estimators as well as the empirical
  coverage of their 95\% confidence intervals for Causal Forest and
  LASSO.} \label{tb:simulationcv}
\end{table}

We obtain similar findings for the cross-fitting case.
Table~\ref{tb:simulationcv} shows the results for Causal Forest and
LASSO.  Unfortunately, BART is too computationally intensive to
include for this simulation.  For the results of Causal Forest and
LASSO, we utilize 1,000 trials as before. As seen in the
sample-splitting case, the estimated biases of the proposed GATES
estimators are relatively small even when $n=100$ and becomes
negligible when $n=500$.

Recall that under the 5-fold cross-fitting, for example, $n=500$
implies the evaluation sample of size 100 for each fold.  And, yet,
combining the five folds leads to a much lower standard deviation than
the sample-splitting case with the $n=100$ case in
Table~\ref{tb:simulation}.  The results are similar when comparing the
$n=2500$ cross-fitting case with the $n=500$ sample-splitting case.
Indeed, in some cases, the reduction in standard deviation is more
than 50 percent.  This experimentally demonstrates the efficiency gain
from using a cross-fitting approach.  We further find that although
Theorem~\ref{thm:GATEestcv} implies that the proposed variance
estimate is conservative, the results show only the slight
overcoverage of the confidence intervals. In \citet{imai:li:23a} we
show that the methodology proposed in \cite{cher:etal:19} leads to
more significant overcoverage of the confidence intervals.

\subsection{Finite-Sample Performance of the Proposed Hypothesis Tests}

We next examine the finite sample performance of the proposed
hypothesis tests.  Due to the aforementioned computational intensity of
BART, we focus on Causal Forest and LASSO.  For each simulated data
set, we conduct hypothesis tests of two null hypotheses of interest:
treatment effect homogeneity (see
Equations~\eqref{eq:H0}~and~\eqref{eq:H0cross} for sample-splitting
and cross-fitting, respectively) and rank-consistency of the GATES
(see Equations~\eqref{eq:H0ast}~and~\eqref{eq:H0cosis-cross} for
sample-splitting and cross-fitting cases, respectively).

According to the true values shown in
Tables~\ref{tb:simulation}~and~\ref{tb:simulationcv}, the null
hypothesis of treatment effect homogeneity is false while the
rank-consistency null hypothesis is correct.  This suggests that the
proposed test should reject the former hypothesis more frequently as
the sample size increases whereas it should reject the latter
hypothesis no more frequently than the specified size of the test,
which we set to 5\% throughout.

\begin{table}[t!]
	\centering\setlength{\tabcolsep}{2.5pt}
	\spacingset{1}
	\begin{tabular}{l|..|..|..}
		\hline
		&  \multicolumn{2}{c|}{$\bm{n_{\text{test}}=100}$} &  \multicolumn{2}{c|}{$\bm{n_{\text{test}}=500}$}& \multicolumn{2}{c}{$\bm{n_{\text{test}}=2500}$}\\\hline
		& \multicolumn{1}{c}{rejection} &
                                                \multicolumn{1}{c|}{median}
		& \multicolumn{1}{c}{rejection} &
                                                \multicolumn{1}{c|}{median}
		& \multicolumn{1}{c}{rejection} &
                                                \multicolumn{1}{c}{median} \\
		& \multicolumn{1}{c}{rate} &
                                                \multicolumn{1}{c|}{$p$-value}
		& \multicolumn{1}{c}{rate} &
                                                \multicolumn{1}{c|}{$p$-value}
		& \multicolumn{1}{c}{rate} &
                                                \multicolumn{1}{c}{$p$-value} \\ \hline
		\multicolumn{1}{l|}{\textbf{Causal Forest}}  & & & & & &\\
		\multicolumn{1}{l|}{$H_0$:\textit{Treatment effect homogeneity}}  & & & &&&\\
		\quad $n_{\text{train}}=100$  & 5.2\% & 0.504  & 7.4\% & 0.529   &  19.6\%  & 0.361   \\
		\quad $n_{\text{train}}=400$ & 9.0  & 0.459  & 22.0 & 0.254    &  74.4 &  0.002   \\
		\quad $n_{\text{train}}=2000$   & 13.0 & 0.367  & 40.4 & 0.092  &  96.0  &  0.000   \\
		\multicolumn{1}{l|}{$H_0^\ast$: \textit{Rank consistency of GATES}}  & & & &&&\\
		\quad $n_{\text{train}}=100$  & 4.0\% & 0.583  & 2.2\% & 0.624  &  2.2\%  & 0.704   \\
		\quad $n_{\text{train}}=400$ & 2.8  & 0.687  & 0.2 & 0.820    &  0.2 &  0.907   \\
		\quad $n_{\text{train}}=2000$   & 1.2 & 0.707  & 0.2 & 0.852  &  0.0  &  0.967   \\\hline
		\multicolumn{1}{l|}{\textbf{LASSO}}  & & & & & &\\
		\multicolumn{1}{l|}{$H_0$: \textit{Treatment effect homogeneity}}  & & & &&&\\
		\quad $n_{\text{train}}=100$  & 5.8\% & 0.496  & 5.2\% & 0.544   &  9.6\%  & 0.516   \\
		\quad $n_{\text{train}}=400$ & 7.0  & 0.557  & 4.0 & 0.578    &  10.4 &  0.468   \\
		\quad $n_{\text{train}}=2000$   & 6.2  & 0.489  & 9.4 & 0.519 &  26.2  &  0.249   \\
		\multicolumn{1}{l|}{$H_0^\ast$: \textit{Rank consistency of GATES}}  & & & &&&\\
		\quad $n_{\text{train}}=100$  & 4.6\% & 0.525  & 3.0\% & 0.584  &  5.4\%  & 0.596   \\
		\quad $n_{\text{train}}=400$ & 6.0  & 0.494  & 1.8 & 0.600    &  2.4 &  0.687   \\
		\quad $n_{\text{train}}=2000$   & 3.2 & 0.608  & 1.4 & 0.698  &  1.2  &  0.851   \\\hline\hline
	\end{tabular}
	\caption{The Finite Sample Performance of the Hypothesis Tests
          for Treatment Effect Homogeneity and Rank-consistency of
          GATES under Sample-splitting.  The results are based on
          Causal Forest and LASSO.  The table presents the percent of
          500 simulation trials where each null hypothesis is rejected
          using the $5\%$ test size.  In addition, the median
          $p$-value across all trials is also shown.  The results are
          presented for different training data sizes
          $n_{\text{train}}$ and different test data sizes
          $n_{\text{test}}$.} \label{tb:synthetichypo}
\end{table}

We first consider the sample-splitting setting based on 500 simulation
trials. Table~\ref{tb:synthetichypo} presents the
rejection rate and median $p$-value for each scenario across different
training and testing data sizes, denoted by $n_{\text{train}}$ and
$n_{\text{test}}$, respectively.  We find that for Causal Forest, the
training data of size 400 and the test data of size 2000 are required
to reject the null hypothesis of treatment effect homogeneity with a
high probability.  This highlights the difficulty of identifying
treatment effect heterogeneity in randomized experiments.  For the
hypothesis test of the rank-consistency of GATES, we find that if
trained with a small sample ($n_{\text{train}}=100$), Causal Forest
might falsely reject the null hypothesis but this false rejection rate
is less than the size of the test regardless of the size of the test
data. This reflects the conservative nature of our test as discussed
at the end of Section~\ref{sec:sample_splitting}.

We obtain similar findings for LASSO, where small training data leads
to low rejection rates for the treatment effect homogeneity hypothesis
and some false rejection of the rank consistency hypothesis.  As
before, the false rejection rates are approximately 5\% or lower (the
small number of simulations induce some noise).  Interestingly, the
proposed test is much less powerful for LASSO than for Causal Forest.
Even when the size of training data is 2,000 and the test data size is
2,500, the rejection rate is only slightly above 25\%.  This is
consistent with the finding in Section~\ref{subsec:GATEest} that LASSO
discovers less treatment effect heterogeneity than Causal Forest.

\begin{table}[t!]
	\centering\setlength{\tabcolsep}{2.5pt}
	\spacingset{1}
	\begin{tabular}{l|..|..|..}
		\hline
		&  \multicolumn{2}{c|}{$\bm{n=100}$} &  \multicolumn{2}{c|}{$\bm{n=500}$}& \multicolumn{2}{c}{$\bm{n=2500}$}\\\hline
		& \multicolumn{1}{c}{rejection} &
		\multicolumn{1}{c|}{median}
		& \multicolumn{1}{c}{rejection} &
		\multicolumn{1}{c|}{median}
		& \multicolumn{1}{c}{rejection} &
		\multicolumn{1}{c}{median} \\
		& \multicolumn{1}{c}{rate} &
		\multicolumn{1}{c|}{$p$-value}
		& \multicolumn{1}{c}{rate} &
		\multicolumn{1}{c|}{$p$-value}
		& \multicolumn{1}{c}{rate} &
		\multicolumn{1}{c}{$p$-value} \\ \hline
		\multicolumn{1}{l|}{\textbf{Causal Forest}}  & & & & & &\\
		\multicolumn{1}{l|}{Homogeneous Treatment Effects}   & 1.4\% & 0.790  & 4.6\% & 0.712   &  51.4\%  & 0.041   \\
		\multicolumn{1}{l|}{Consistent Treatment Effects}  & 1.4\% & 0.702  & 0.8\% & 0.845  &  0.0\%  & 0.976   \\\hline
		\multicolumn{1}{l|}{\textbf{LASSO}}  & & & & & &\\
		\multicolumn{1}{l|}{Homogeneous Treatment Effects}   & 0.6\% & 0.880  & 1.8\% & 0.850   &  9.0\%  & 0.664   \\
		\multicolumn{1}{l|}{Consistent Treatment Effects}  & 1.0\% & 0.722  & 0.6\% & 0.769  &  0.2\%  & 0.889   \\\hline\hline
	\end{tabular}
	\caption{The Finite Sample Performance of the Hypothesis Tests
          for Treatment Effect Homogeneity and Rank-consistency of
          GATES under Cross-fitting.  The results are based on Causal
          Forest and LASSO.  The table presents the percent of 500
          simulation trials where each null hypothesis is rejected
          using the $5\%$ test size and also the median $p$-value
          across all trials.} \label{tb:synthetichypocv}
\end{table}

We also examine the performance of our hypothesis tests under
cross-fitting, again using 500 simulation
trials. Table~\ref{tb:synthetichypocv} presents the rejection rate and
median $p$-value across different sample sizes.  We use $L=5$ fold
cross-fitting for all simulations.  Note that the $n=500$ case under
cross-fitting is analogous in the size of training and testing data to
the $(n_{\text{train}}=400,n_{\text{test}}=100)$ case for sample
splitting.  Similarly, the $n=2500$ case under cross-fitting
corresponds to the $(n_{\text{train}}=2,000,n_{\text{test}}=500)$ case
under sample-splitting.

For both Causal Forest and LASSO, the rejection rate of the
homogeneous treatment effect hypothesis is lower in the $n=500$ case
compared with the corresponding sample-splitting case, reflecting the
additional uncertainty due to the sampling of training data (under
sample-splitting, the scoring rule is regarded as fixed). However,
when the sample size is $n=2,500$, for both algorithms the rejection
rate of homogeneous treatment effects is higher under cross-fitting
than sample-splitting, demonstrating that the efficiency gain of
cross-fitting outweigh its additional sampling uncertainty. For the
hypothesis test of rank-consistency, we find that the rejection rate
under cross-fitting is significantly lower than the nominal test size
for all cases.

\section{An Empirical Application}
\label{sec:empirical}

To demonstrate the applicability of the proposed framework, we utilize
the experimental data from the male sub-sample of the National
Supported Work Demonstration (NSW)
\citep{lalonde1986evaluating,dehejia1999causal}. NSW was a temporary
employment program to help disadvantaged workers by providing them
with work experience and counseling in a sheltered
environment. Specifically, qualified applicants were randomly assigned
to the treatment and control groups, where the workers in the
treatment group were given a guaranteed job for 9 to 18 months. The
primary outcome of interest is the annualized earnings in 1978, 36
months after the program.  The data contains a total of $n=722$
observations, with $n_1=297$ participants assigned to the treatment
group and $n_0=425$ participants in the control group. There are 7
available pre-treatment covariates $\bm{X}$ that records the
demographics and pre-treatment earnings of the participants.

We evaluate Causal Forest, BART, and LASSO under the two settings
considered in this paper.  For sample-splitting, we randomly select
$67\%$ of the data (484 observations) to serve as a training dataset.
We use the remaining 238 samples to estimate the GATES and conduct the
proposed hypothesis tests.  For cross-fitting, we first randomly split
the data into 3 folds, i.e., $L=3$.  We use each fold once as a
testing set, while the remaining two folds are the training set. The
number of trees was tuned through 5-fold cross-validation for BART and
Causal Forest within each training dataset. The regularization
parameter was tuned similarly for LASSO.

\begin{table}[t!]
	\centering\setlength{\tabcolsep}{6pt}
	\small \spacingset{1}
	\begin{tabular}{l|C|C|C|C|C}
		\hline
			& \multicolumn{1}{c|}{$\bm{\hat\tau_1}$} &\multicolumn{1}{c|}{$\bm{\hat\tau_2}$}& \multicolumn{1}{c|}{$\bm{\hat\tau_3}$}&\multicolumn{1}{c|}{$\bm{\hat\tau_4}$}& \multicolumn{1}{c}{$\bm{\hat\tau_5}$}\\\hline
		\multicolumn{1}{l|}{\textbf{Sample-splitting}}  & & & & & \\
		\quad Causal Forest  & 3.40   & 0.13  & -0.85  & -1.91 & 7.21 \\
		 & [-1.29,3.40]  & [-5.37,5.63] & [-5.22, 3.52]&[-5.16,1.34]&[1.22,13.19]\\
		\quad BART  & 2.90  & -0.73  & -0.02 & 3.25 & 2.57 \\
   & [-2.25,8.06]  & [-5.05,3.58] & [-3.47,3.43]&[-1.53,8.03]&[-3.82,8.97]\\
		\quad LASSO  & 1.86  & 2.62  & -2.07 & 1.39 & 4.17 \\
& [-3.59, 7.30]  & [-1.69,6.93] & [-5.39,1.26]&[-2.95,5.73]&[-2.30,10.65]\\
\multicolumn{1}{l|}{\textbf{Cross-fitting}}  & & & & &\\
		\quad Causal Forest  & -3.72 & 1.05  & 5.32  &-2.64& 4.55 \\
		  & [-6.52,-0.93]  & [-2.28,4.37] & [2.63,8.01]&[-5.07,-0.22]&[1.14,7.96]\\
\quad BART  & 0.40  & -0.15  & -0.40 & 2.52 & 2.19\\
& [-3.79,4.59]  & [-2.54,2.23] & [-3.37,2.56]&[-0.99,6.03]&[-0.73,5.11]\\
		\quad LASSO  & 0.65  & 0.45  & -2.88 & 1.32 & 5.02 \\
& [-3.65,4.94]  & [-3.28,4.18] & [-5.38,-0.38]&[-1.83,4.48]&[-0.14,10.18]\\\hline\hline
	\end{tabular}
	\caption{The Estimated GATES and their $95\%$ Confidence
          Intervals based on Causal Forest, BART, and LASSO under
          Sample-splitting and Cross-fitting.  The estimated GATES
          based on quintiles are reported in 1,000 US dollars.
          Sample-splitting is done using 67\% of the sample as the
          training data and 33\% of the sample as the evaluation data.
          For cross-fitting, we use 3 folds of equal
          size.} \label{tb:realgates}
\end{table}

We focus on the quintile GATES ($K=5$). Table~\ref{tb:realgates}
presents the results (reported in 1,000 US dollars) under the
sample-splitting and cross-fitting settings.  We find that Causal
Forest is able to produce statistically significantly positive GATES
for the highest quintile group ($\hat{\tau}_5$) under both
sample-splitting and cross-fitting. Thus, unlike the other two
algorithms, Causal Forest can identify a $20\%$ subset that benefits
significantly from the temporary employment program.

Two additional observations are worth noting.  First, the confidence
intervals are generally narrower in the cross-fitting case compared to
the sample-splitting case.  This finding is consistent with the fact
that cross-fitting is more efficient than sample-splitting.  Second,
the three algorithms failed to produce any statistically significant
positive GATES for the remaining groups.  This may be because there are
few additional workers who benefit from the program.  Alternatively,
it is also possible that such workers exist but the algorithms are
unable to identify them.

\begin{table}[t!]
	\centering\setlength{\tabcolsep}{7pt}
	\small \spacingset{1}
	\begin{tabular}{l|..|..|..}
		\hline
			& \multicolumn{2}{c|}{\textbf{Causal Forest}} &\multicolumn{2}{c|}{\textbf{BART}}&\multicolumn{2}{c}{\textbf{LASSO}}\\
		& \multicolumn{1}{c}{\text{stat}} &  \multicolumn{1}{c|}{\text{$p$-value} }	& \multicolumn{1}{c}{\text{stat}} &  \multicolumn{1}{c|}{\text{$p$-value} }	& \multicolumn{1}{c}{\text{stat}} &  \multicolumn{1}{c}{\text{$p$-value} }	  \\\hline
		\multicolumn{1}{l|}{\textbf{Sample-splitting}}  & & & &  & &\\
		\multicolumn{1}{l|}{Homogeneous Treatment Effects}   &9.78   & 0.082  &2.76  & 0.737 &5.26  & 0.362 \\
		\multicolumn{1}{l|}{Rank-consistent Treatment Effects}  & 3.07   & 0.323 & 1.13 & 0.657 & 3.14 & 0.302\\
		\multicolumn{1}{l|}{\textbf{Cross-fitting}}  & & & & & &  \\
\multicolumn{1}{l|}{Homogeneous Treatment Effects}   &  30.29   & 0.000  & 2.32  & 0.803 & 10.79 & 0.056\\
\multicolumn{1}{l|}{Rank-consistent Treatment Effects}  & 0.06   & 0.691  & 0.04  & 0.885 & 0.45  & 0.711\\
		\hline\hline
	\end{tabular}
	\caption{The Results of the Proposed Hypothesis Tests under
          Sample-splitting and Cross-fitting Using Causal Forest,
          BART, and LASSO.  The values of test statistics and
          $p$-values are presented. We test the null hypotheses of
          treatment effect homogeneity and rank-consistency of the
          GATES.} \label{tb:realtests}
\end{table}

To formally evaluate the statistical significance of several GATES
estimates, we must account for the potential multiple testing problem.
Thus, we apply the proposed hypothesis tests to evaluate the null
hypotheses of treatment effect homogeneity and rank-consistency of the
GATES.  Table~\ref{tb:realtests} presents the resulting values of test
statistics and $p$-values.  We find that under sample-splitting, only
Causal Forest is able to reject the null hypothesis of treatment
effect homogeneity at the 10\% level. However, under cross-fitting,
both Causal Forest and LASSO algorithms can reject the null hypothesis
at the $10\%$ level, with Causal Forest being able to reject the
hypothesis at even the $0.1\%$ level. In contrast, BART fails to
reject the treatment effect homogeneity hypothesis under both
sample-splitting and cross-fitting.  The results with Causal Forest
suggest that the identification of a statistically significant GATES
estimate for one subgroup under cross-fitting is able to grant enough
power to reject the null hypothesis that the average treatment effects
are homogeneous across all subgroups.  Finally, we find that all three
algorithms fail to reject the null hypothesis of the rank-consistency
of GATEs.  Thus, under our conservative tests, there is no strong
statistical evidence that these algorithms are producing unreliable
GATES.

\section{Concluding Remarks}
\label{sec:conclude}

Many randomized experiments have a limited sample size and the
resulting treatment effect estimates are often small and noisy.  Yet,
applied researchers often use machine learning algorithms to estimate
heterogeneous treatment effects.  Therefore, it is important to
statistically distinguish signal from noise.  We have developed the
framework that does not impose a strong assumption on machine learning
algorithms and hence is applicable to a wide range of situations.  The
proposed methodology allows researchers to construct confidence
intervals on the estimated average treatment effects within a group
identified by any machine learning algorithm.  We also show how to
conduct formal hypothesis tests about heterogeneous treatment effects.
Our method solely relies upon the randomization of treatment
assignment and the random sampling of units, and hence yields reliable
statistical inference even when the sample size is relatively small
and machine learning algorithms are not performing well.

\bigskip
\pdfbookmark[1]{References}{References}
\bibliography{sample,my,imai}

\begin{thebibliography}{}

\bibitem[Andrews and Guggenberger(2009)]{andrews2009validity}
Andrews, D.~W. and Guggenberger, P. (2009).
\newblock Validity of subsampling and “plug-in asymptotic” inference for
  parameters defined by moment inequalities.
\newblock \emph{Econometric Theory} \textbf{25}, 3, 669--709.

\bibitem[Andrews and Soares(2010)]{andrews2010inference}
Andrews, D.~W. and Soares, G. (2010).
\newblock Inference for parameters defined by moment inequalities using
  generalized moment selection.
\newblock \emph{Econometrica} \textbf{78}, 1, 119--157.

\bibitem[Austern and Zhou(2020)]{austern2020asymptotics}
Austern, M. and Zhou, W. (2020).
\newblock Asymptotics of cross-validation.
\newblock \emph{arXiv preprint arXiv:2001.11111} .

\bibitem[Bhattacharya(1974)]{bhattacharya1974convergence}
Bhattacharya, P.~K. (1974).
\newblock Convergence of sample paths of normalized sums of induced order
  statistics.
\newblock \emph{The Annals of Statistics}  1034--1039.

\bibitem[Canay \emph{et~al.}(2023)Canay, Illanes, and Velez]{canay2023user}
Canay, I., Illanes, G., and Velez, A. (2023).
\newblock A user's guide to inference in models defined by moment inequalities.
\newblock Tech. rep., National Bureau of Economic Research.

\bibitem[Canay(2010)]{canay2010inference}
Canay, I.~A. (2010).
\newblock El inference for partially identified models: Large deviations
  optimality and bootstrap validity.
\newblock \emph{Journal of Econometrics} \textbf{156}, 2, 408--425.

\bibitem[Chernozhukov \emph{et~al.}(2019)Chernozhukov, Chetverikov, and
  Kato]{chernozhukov2019inference}
Chernozhukov, V., Chetverikov, D., and Kato, K. (2019).
\newblock Inference on causal and structural parameters using many moment
  inequalities.
\newblock \emph{The Review of Economic Studies} \textbf{86}, 5, 1867--1900.

\bibitem[Chernozhukov \emph{et~al.}(2023)Chernozhukov, Demirer, Duflo, and
  {Fernandez-Val}]{cher:etal:19}
Chernozhukov, V., Demirer, M., Duflo, E., and {Fernandez-Val}, I. (2023).
\newblock Fisher-schultz lecture: Generic machine learning inference on
  heterogeneous treatment effects in randomized experiments.
\newblock Tech. rep., arXiv:1712.04802.

\bibitem[Chipman \emph{et~al.}(2010)Chipman, George, McCulloch,
  \emph{et~al.}]{chipman2010bart}
Chipman, H.~A., George, E.~I., McCulloch, R.~E., \emph{et~al.} (2010).
\newblock Bart: Bayesian additive regression trees.
\newblock \emph{The Annals of Applied Statistics} \textbf{4}, 1, 266--298.

\bibitem[Cram{\'e}r and Wold(1936)]{cramer1936some}
Cram{\'e}r, H. and Wold, H. (1936).
\newblock Some theorems on distribution functions.
\newblock \emph{Journal of the London Mathematical Society} \textbf{1}, 4,
  290--294.

\bibitem[Crump \emph{et~al.}(2008)Crump, Hotz, Imbens, and
  Mitnik]{crum:etal:08}
Crump, R.~K., Hotz, V.~J., Imbens, G.~W., and Mitnik, O.~A. (2008).
\newblock Nonparametric tests for treatment effect heterogeneity.
\newblock \emph{The Review of Economics and Statistics} \textbf{90}, 3,
  389--405.

\bibitem[Dehejia and Wahba(1999)]{dehejia1999causal}
Dehejia, R.~H. and Wahba, S. (1999).
\newblock Causal effects in nonexperimental studies: Reevaluating the
  evaluation of training programs.
\newblock \emph{Journal of the American statistical Association} \textbf{94},
  448, 1053--1062.

\bibitem[Ding \emph{et~al.}(2016)Ding, Feller, and Miratrix]{ding:fell:mira:16}
Ding, P., Feller, A., and Miratrix, L. (2016).
\newblock Randomization inference for treatment effect variation.
\newblock \emph{Journal of the Royal Statistical Society, Series {B}
  (Statistical Methodology)} \textbf{78}, 3, 655--671.

\bibitem[Ding \emph{et~al.}(2019)Ding, Feller, and Miratrix]{ding:fell:mira:19}
Ding, P., Feller, A., and Miratrix, L. (2019).
\newblock Decomposing treatment effect variation.
\newblock \emph{Journal of the American Statistical Association} \textbf{114},
  525, 304--317.

\bibitem[Dorie \emph{et~al.}(2019)Dorie, Hill, Shalit, Scott, and
  Cervone]{dori:etal:19}
Dorie, V., Hill, J., Shalit, U., Scott, M., and Cervone, D. (2019).
\newblock Automated versus do-it-yourself methods for causal inference: Lessons
  learned from a data analysis competition.
\newblock \emph{Statistical Science} \textbf{34}, 1, 43--68.

\bibitem[Dwivedi \emph{et~al.}(2020)Dwivedi, Tan, Park, Wei, Horgan, Madigan,
  and Yu]{dwivedi2020stable}
Dwivedi, R., Tan, Y.~S., Park, B., Wei, M., Horgan, K., Madigan, D., and Yu, B.
  (2020).
\newblock Stable discovery of interpretable subgroups via calibration in causal
  studies.
\newblock \emph{International Statistical Review} \textbf{88}, S135--S178.

\bibitem[Hahn \emph{et~al.}(2020)Hahn, Murray, Carvalho,
  \emph{et~al.}]{hahn2020bayesian}
Hahn, P.~R., Murray, J.~S., Carvalho, C.~M., \emph{et~al.} (2020).
\newblock Bayesian regression tree models for causal inference: regularization,
  confounding, and heterogeneous effects.
\newblock \emph{Bayesian Analysis} \textbf{15}, 3, 965--1056.

\bibitem[Higham(2002)]{higham2002computing}
Higham, N.~J. (2002).
\newblock Computing the nearest correlation matrix—a problem from finance.
\newblock \emph{IMA journal of Numerical Analysis} \textbf{22}, 3, 329--343.

\bibitem[Hill(2011)]{hill2011bayesian}
Hill, J.~L. (2011).
\newblock Bayesian nonparametric modeling for causal inference.
\newblock \emph{Journal of Computational and Graphical Statistics} \textbf{20},
  1, 217--240.

\bibitem[Holland(1986)]{holl:86}
Holland, P.~W. (1986).
\newblock Statistics and causal inference (with discussion).
\newblock \emph{Journal of the American Statistical Association} \textbf{81},
  945--960.

\bibitem[Imai and Li(2023{a})]{imai:li:23a}
Imai, K. and Li, M.~L. (2023{a}).
\newblock Comment on ``generic machine learning inference on heterogeneous
  treatment effects in randomized experiments.''.
\newblock Tech. rep., Working Paper.

\bibitem[Imai and Li(2023{b})]{imai2019experimental}
Imai, K. and Li, M.~L. (2023{b}).
\newblock Experimental evaluation of individualized treatment rules.
\newblock \emph{Journal of the American Statistical Association} \textbf{118},
  541, 242--256.

\bibitem[Imai and Li(2023{c})]{imai:li:23b}
Imai, K. and Li, M.~L. (2023{c}).
\newblock Statistical performance guarantee for selecting those predicted to
  benefit most from treatment.
\newblock Tech. rep., arXiv preprint 2310.07973.

\bibitem[K\"{u}nzel \emph{et~al.}(2018)K\"{u}nzel, Sekhon, Bickel, and
  Yu]{kunz:etal:18}
K\"{u}nzel, S.~R., Sekhon, J.~S., Bickel, P.~J., and Yu, B. (2018).
\newblock Meta-learners for estimating heterogeneous treatment effects using
  machine learning.
\newblock Tech. rep., arXiv:1706.03461.

\bibitem[LaLonde(1986)]{lalonde1986evaluating}
LaLonde, R.~J. (1986).
\newblock Evaluating the econometric evaluations of training programs with
  experimental data.
\newblock \emph{The American economic review}  604--620.

\bibitem[Nadeau and Bengio(2000)]{nadeau2000inference}
Nadeau, C. and Bengio, Y. (2000).
\newblock Inference for the generalization error.
\newblock In \emph{Advances in neural information processing systems},
  307--313.

\bibitem[Neyman(1923)]{neym:23}
Neyman, J. (1923).
\newblock On the application of probability theory to agricultural experiments:
  Essay on principles, section 9. (translated in 1990).
\newblock \emph{Statistical Science} \textbf{5}, 465--480.

\bibitem[Rubin(1990)]{rubi:90}
Rubin, D.~B. (1990).
\newblock Comments on ``{O}n the application of probability theory to
  agricultural experiments. {E}ssay on principles. {S}ection 9'' by {J.}
  {S}plawa-{N}eyman translated from the {P}olish and edited by {D.} {M.}
  {D}abrowska and {T.} {P.} {S}peed.
\newblock \emph{Statistical Science} \textbf{5}, 472--480.

\bibitem[Shapiro(1988)]{shapiro1988towards}
Shapiro, A. (1988).
\newblock Towards a unified theory of inequality constrained testing in
  multivariate analysis.
\newblock \emph{International Statistical Review/Revue Internationale de
  Statistique} \textbf{56}, 1, 49--62.

\bibitem[Shorack(1972)]{shorack1972functions}
Shorack, G.~R. (1972).
\newblock Functions of order statistics.
\newblock \emph{The Annals of Mathematical Statistics} \textbf{43}, 2,
  412--427.

\bibitem[Tibshirani(1996)]{tibs:96}
Tibshirani, R. (1996).
\newblock Regression shrinkage and selection via {LASSO}.
\newblock \emph{Journal of the Royal Statistical Society, Series {B}
  (Statistical Methodology)} \textbf{58}, 1, 267--288.

\bibitem[Wager and Athey(2018)]{wage:athe:18}
Wager, S. and Athey, S. (2018).
\newblock Estimation and inference of heterogeneous treatment effects using
  random forests.
\newblock \emph{Journal of the American Statistical Association} \textbf{113},
  523, 1228--1242.

\bibitem[Wellner(1977)]{wellner1977glivenko}
Wellner, J.~A. (1977).
\newblock A glivenko-cantelli theorem and strong laws of large numbers for
  functions of order statistics.
\newblock \emph{The Annals of Statistics} \textbf{5}, 3, 473--480.

\bibitem[Yadlowsky \emph{et~al.}(2021)Yadlowsky, Fleming, Shah, Brunskill, and
  Wager]{yadlowsky2021evaluating}
Yadlowsky, S., Fleming, S., Shah, N., Brunskill, E., and Wager, S. (2021).
\newblock Evaluating treatment prioritization rules via rank-weighted average
  treatment effects.
\newblock \emph{arXiv preprint arXiv:2111.07966} .

\end{thebibliography}

\clearpage
\appendix
\spacingset{1}

\setcounter{table}{0}
\renewcommand{\thetable}{S\arabic{table}}
\setcounter{figure}{0}
\renewcommand{\thefigure}{S\arabic{figure}}
\setcounter{equation}{0}
\renewcommand{\theequation}{S\arabic{equation}}
\setcounter{theorem}{0}
\renewcommand {\thetheorem} {S\arabic{theorem}}
\setcounter{section}{0}
\renewcommand {\thesection} {S\arabic{section}}
\setcounter{lemma}{0}
\renewcommand {\thelemma} {S\arabic{lemma}}

\begin{center}
\LARGE {\bf Supplementary Appendix}
\end{center}

\section{Proof of Theorem \ref{thm:GATEest}}
\label{app:GATEest}


We first rewrite the expectation of the proposed estimator in
Equation~\eqref{eq:gCATEest} as,
\begin{equation*}
	\E(\hat \tau_k) \ = \ K \E \l\{Y_i\l(f^\ast(\bX_i,\hat{c}_k(s))\r) - Y_i\l(f^\ast(\bX_i,\hat{c}_{k-1}(s))\r)\r\},
\end{equation*}
where $f^\ast(\bX_i,c)=\bone\{s(\bX_i) < c\}$.  Similarly, we can also
write the estimand in Equation~\eqref{eq:gCATE} as,
\begin{equation*}
	\tau_k \ = \ K \E\l\{Y_i\l(f^\ast(\bX_i,c_k(s))\r) - Y_i\l(f^\ast(\bX_i,c_{k-1}(s))\r)\r\}.
\end{equation*}
Now, define $F(c)=\mc{P}(s(\bX_i)\leq c)$. Without loss of generality,
assume $\hat{c}_k(s) > c_k(s)$ and $\hat{c}_{k-1}(s) > c_{k-1}(s)$.
If this is not the case, we simply switch the upper and lower limits
of the integrals in the proof below.  Then, the bias of the estimator
is given by,
\begin{eqnarray*}
  & & \frac{\l|\E(\hat \tau_k)-\tau_k\r|}{K} \\
  &\leq & \l|\E
                                            \l\{Y_i\l(f^\ast(\bX_i,\hat{c}_k(s))\r)-Y_i\l(f(\bX_i,c_k(s))\r)\r\}\r| +  \l|\E\l\{Y_i\l(f^\ast(\bX_i,\hat{c}_{k-1}(s))\r)-Y_i\l(f^\ast(\bX_i,c_{k-1}(s))\r)\r\}\r|\\
	& = & \l|\E_{\hat{c}_k(s)}\l[\int^{\hat{c}_k(s)}_{c_k(s)}
	\E(\tau_i \mid s(\bX_i)=c) \d F(c)\r]\r| +  \l|\E_{\hat{c}_{k-1}(s)}\l[\int^{\hat{c}_{k-1}(s)}_{c_{k-1}(s)}
\E(\tau_i \mid s(\bX_i)=c) \d F(c)\r]\r|\\
	& = & \l|\E_{F(\hat{c}_k(s))}\l[\int^{F(\hat{c}_k(s))}_{F(c_k(s))}
	\E(\tau_i\mid s(\bX_i)=F^{-1}(x)) \d x\r]\r|\\
	&  & + \l|\E_{F(\hat{c}_{k-1}(s))}\l[\int^{F(\hat{c}_{k-1}(s))}_{F(c_{k-1}(s))}
\E(\tau_i\mid s(\bX_i)=F^{-1}(x)) \d x\r]\r|\\
	& \leq & \E_{F(\hat{c}_k(s))}\l[\l|F(\hat{c}_k(s))-\frac{k}{K}\r| \times \max_{c \in
		[c_k(s),\hat{c}_k(s)]} \l | \E(\tau_i\mid s(\bX_i)=c)\r|\r]\\&  & +\E_{F(\hat{c}_{k-1}(s))}\l[\l|F(\hat{c}_{k-1}(s))-\frac{k-1}{K}\r| \times \max_{c \in
		[c_{k-1}(s),\hat{c}_{k-1}(s)]} \l | \E(\tau_i\mid s(\bX_i)=c)\r|\r]
\end{eqnarray*}
By the definition of $\hat{c}_k(s)$, $F(\hat{c}_k(s))$ is the $nk/K$th
order statistic of $n$ independent uniform random variables, and thus
follows the Beta distribution with the shape and scale parameters
equal to $nk/K$ and $n-nk/K+1$, respectively. For the special case
where $k-1=0$, we define the $0$th order statistic of $n$ uniform
random variables to be 0, and by extension also define the ``beta
distribution'' with shape parameter $\leq 0$ to be $H(x)$ where $H(x)$
is the Heaviside step function. Therefore, we have,
\begin{equation*}
	\mc{P}\l(|F(\hat{c}_k(s))-\frac{k}{K}|>\epsilon\r) \ = \
        1-B\l(\frac{k}{K}+\epsilon, \frac{nk}{K},
        n-\frac{nk}{K}+1\r)+B\l(\frac{k}{K}-\epsilon, \frac{nk}{K},
        n-\frac{nk}{K}+1\r),
\end{equation*}
where
$B(\epsilon,\alpha,\beta) \ = \ \int^\epsilon_0 t^{\alpha -1}
(1-t)^{\beta -1} \d t$ is the incomplete beta function.  Similarly, we
have
\begin{eqnarray*}
 \mc{P}(|F(\hat{c}_{k-1}(s))-\frac{k-1}{K}|>\epsilon)
  & = & 1-B\l(\frac{k-1}{K}+\epsilon, \frac{n(k-1)}{K},
        n-\frac{n(k-1)}{K}+1\r)\\
  & & +B\l(\frac{(k-1)}{K}-\epsilon, \frac{n(k-1)}{K}, n-\frac{n(k-1)}{K}+1\r).
\end{eqnarray*}
Combining the above results yields the desired bias bound expression.

To derive the exact variance, we first apply the law of total variance
to Equation~\eqref{eq:gCATEest},
\begin{align}
	\V(\hat\tau_k) \  = & \
                           \V\left[\E\left\{K\left(\frac{1}{n_1}\sum_{i=1}^n
                           \hat{f}_k(\bX_i)T_i Y_i(1)-\frac{1}{n_0}\sum_{i=1}^n\hat{f}_k(\bX_i)(1-T_i) Y_i(0)
                           \right)\biggm\vert \bX,
                           \{Y_i(1), Y_i(0)\}_{i=1}^n \right\}\right]\nonumber\\
                         &+\E\left[\V\left\{K\left(\frac{1}{n_1}\sum_{i=1}^n
                          \hat{f}_k(\bX_i)T_i Y_i(1)-\frac{1}{n_0}\sum_{i=1}^n \hat{f}_k(\bX_i)(1-T_i)Y_i(0)
                           \right)\biggm\vert \bX, \{Y_i(1), Y_i(0)\}_{i=1}^n  \right\}\right]\nonumber
  \\
	= & \ K^2\V\l(\frac{1}{n} \sum_{i=1}^n \{Y_{ki}(1) -
            Y_{ki}(0)\}\r) \nonumber \\
  & +K^2\E\left[\V\left\{\frac{1}{n_1}\sum_{i=1}^n \hat{f}_k(\bX_i)T_iY_i(1)-\frac{1}{n_0}\sum_{i=1}^n\hat{f}_k(\bX_i)(1-T_i)Y_i(0) \biggm\vert \bX, \{Y_i(1), Y_i(0)\}_{i=1}^n\right\}\right]. \label{eq:var_decomp}
\end{align}
Applying the standard result from Neyman's finite sample variance
analysis to the second term shows that this term is equal to,
\begin{align}
  K^2\E\left\{\frac{1}{n} \l( \frac{n_0}{n_1} S_{k1}^2 + \frac{n_1}{n_0} S_{k0}^2
+ 2 S_{k01}\r)\r\}.\label{eq:var2_decomp}
\end{align}
where
$S_{k01} = \sum_{i=1}^n (Y_{ki}(0) - \overline{Y_{k}(0)})(Y_{ki}(1) -
\overline{Y_{k}(1)})/(n-1)$.  Since $Y_{ki}(t)$ and $Y_{kj}(t)$ are
correlated, we apply Lemma~1 of \citet{nadeau2000inference} to the
first term, yielding,
\begin{align}
	& \V\l(\frac{1}{n} \sum_{i=1}^n \{Y_{ki}(1) - Y_{ki}(0)\}\r)
   \nonumber \\
	\ = \ & {\rm Cov}(Y_{ki}(1) -Y_{ki}(0), Y_{kj}(1) -
Y_{kj}(0)) + \frac{1}{n}   \E(S_{k1}^2  + S_{k0}^2 - 2 S_{k01}), \label{eq:var1_decomp}
\end{align}
for $i \ne j$ where
\begin{align*}
	&{\rm Cov}(Y_{ki}(1) -Y_{ki}(0), Y_{kj}(1) -
	Y_{kj}(0))\\
	= \ &  {\rm Cov}\left(\hat{f}_k(\bX_i)\tau_i, \hat{f}_k(\bX_j)\tau_j\right) \nonumber\\
	= \ &  \Pr(\hat{f}_k(\bX_i)=\hat{f}_k(\bX_j)=1)\E[\tau_i\tau_j \mid \hat{f}_k(\bX_i)=\hat{f}_k(\bX_j)=1] - \Pr(\hat{f}_k(\bX_i)=1)^2\E[\tau_i\mid \hat{f}_k(\bX_i)=1]^2 \nonumber\\
	= \ & \frac{n
		-K}{K^2(n-1)}\E[\tau_i\tau_j \mid \hat{f}_k(\bX_i)=\hat{f}_k(\bX_j)=1]
	-  \frac{1}{K^2}\E[\tau_i\mid
	\hat{f}_k(\bX_i)=1]^2\nonumber\\
	= \ &  \frac{(n-K)\kappa_{k11}}{K^2(n-1)} -	\frac{\kappa_{k1}^2}{K^2}
\end{align*}
Substituting
Equations~\eqref{eq:var2_decomp}~and~\eqref{eq:var1_decomp} into
Equation \ref{eq:var_decomp}, we obtain the desired variance
expression.  \qed

\section{Derivation of $\hat{\kappa}_{ktt}$}
\label{app:kappa}

We first rewrite $\kappa_{ktt}$ as:
\begin{equation*}
\kappa_{ktt}= \sum_{u,v \in \{0,1\}}(-1)^{u+v}\E[Y_i(u)Y_j(v) \mid
\hat{f}_k(\bX_i)=\hat{f}_k(\bX_j)=t].
\end{equation*}
We can estimate each conditional expectation term inside of the
summation using its sample analogue:
\begin{align*}
\frac{\sum_{i=1}^n\sum_{j \neq i}   \mathbf{1}\{\hat{f}_k(\bX_i) = \hat{f}_k(\bX_j) =  t\}\{1-u + (2u-1)T_i\}\{1-v + (2v-1)T_j\}Y_iY_j}{\sum_{i=1}^n\sum_{j \neq i}   \mathbf{1}\{\hat{f}_k(\bX_i) = \hat{f}_k(\bX_j) =  t\}\{1-u + (2u-1)T_i\}\{1-v + (2v-1)T_j\}}
\end{align*}
We can further simplify the computation by rewriting the numerator as:
\begin{align*}	
        &\left[\sum_{i=1}^n   \mathbf{1}\{\hat{f}_k(\bX_i) =  t\}\{1-u
          + (2u-1)T_i\}Y_i\right]
          \left[\sum_{i=1}^n   \mathbf{1}\{\hat{f}_k(\bX_i) =  t\}\{1-v
          + (2v-1)T_i\}Y_i\right]\\
  -&\sum_{i=1}^n   \mathbf{1}\{\hat{f}_k(\bX_i) =  t\}\{1-u + (2u-1)T_i\}\{1-v + (2v-1)T_i\}Y_i^2.
\end{align*}
Similarly, we can rewrite the denominator as follows:
\begin{align*}
        &\left[\sum_{i=1}^n   \mathbf{1}\{\hat{f}_k(\bX_i) =  t\}(1-u
          + (2u-1)T_i)\right]
          \left[\sum_{i=1}^n   \mathbf{1}\{\hat{f}_k(\bX_i) =  t\}
          \{1-v + (2v-1)T_i\}\right]\\-&\sum_{i=1}^n
                                         \mathbf{1}\{\hat{f}_k(\bX_i)
                                         =  t\} \{1-u + (2u-1)T_i\}\{1-v + (2v-1)T_i\}.
\end{align*}
Putting these terms together, we obtain the expression of
$\hat\kappa_{ktt}$ given in Section~\ref{subsec:GATE}.

\section{Proof of Theorem \ref{thm:GATEsamp}}
\label{app:GATEsamp}

Given a tuple of $n$ samples $\{Y_i,T_i,\bX_i\}_{i=1}^n$, we first
reorder the sample to $(Y_{[i,n]},T_{[i,n]},\bX_{[i,n]})$ based on the
magnitude of the scoring rule, such that
\[s(\bX_{[1,n]})\leq s(\bX_{[2,n]})\leq \cdots \leq s(\bX_{[n,n]})\]
Then, the proposed GATES estimator can be rewritten as
\begin{equation}
	\hat{\tau}_k = \frac{1}{n}\sum_{i=1}^n \mathbf{1}\left\{\frac{(k-1)n}{K} < i \leq \frac{kn}{K}\right\} U_{[i,n]}
\end{equation}
where
\begin{align}
	U_{[i,n]}& :=KY_{[i,n]}\left(\frac{T_{[i,n]}}{q}-\frac{1-T_{[i,n]}}{1-q}\right),\label{eq:unitdef}
\end{align}
where $q=n_1/n$.  Now, we prove the following two lemmas.

\begin{lemma}\label{lem:conditional_converge}
  Let $(X_1,Y_1), (X_2,Y_2), \cdots$ be a sequence of random vectors.
  For each $n\geq 1$, $(X_1,Y_1),\cdots, (X_n,Y_n)$ possesses a joint
  distribution. Let $\bZ_n=((X_1,Y_1),\cdots, (X_n, Y_n))$ and
  $\bX_n=(X_1,\cdots, X_n)$, and let $W_n(\bZ_n)$ and $S_n(\bX_n)$ be
  measurable vector-valued functions of $\bZ_n$ and $\bX_n$
  respectively. Suppose $S_n(\bX_n)$ converges in distribution to
  $F_S$ and the conditional distribution $W_n(\bZ_n) \mid \bX_n$
  converges in distribution to $F_W$ in probability, where $F_W$ does not depend on
  $\bX_n$. Then, we have that:
	\[(W_n(\bZ_n), S_n(\bX_n)) \to F_WF_S\]
\end{lemma}
\begin{proof}
  The characteristic function of the joint distribution of
  $(W_n(\bZ_n), S_n(\bX_n))$ can be written as:
  \[\varphi_{W_n S_n}(t_1,t_2)=\E[\exp\{i(t_1W_n(\bZ_n)+t_2S_n(\bX_n))\}]\]
  Let $W \sim F_W$ and $S \sim F_S$. Then the characteristic function
  of $F_WF_S$ can be written as:
  \[\varphi_{W S}(t_1,t_2)=\E[\exp\{i(t_1W)\}]\E[\exp\{i(t_2S)\}]\]
  We then have:
  \begin{align*}
    &|\varphi_{W_n S_n}(t_1,t_2)-\varphi_{W
      S}(t_1,t_2)|\\
    = \ &|\E[\exp\{i(t_1W_n(\bZ_n)+t_2S_n(\bX_n))\}]-\E[\exp\{i(t_1W)\}]\E[\exp\{i(t_2S)\}]| \\
    \leq \ & |\E[\E[\exp\{i(t_1W_n(\bZ_n)+t_2S_n(\bX_n))\}\mid \bX_n
             ]]-\E[\exp\{i(t_1W)\}]\E[\exp\{i(t_2S_n(\bX_n))\}]|\\
    &
      +|\E[\exp\{i(t_1W)\}]\E[\exp\{i(t_2S_n(\bX_n))\}]-\E[\exp\{i(t_1W)\}]\E[\exp\{i(t_2S)\}]| \\
        \leq &  \E[|\E[\exp(it_1W_n(\bZ_n)) \mid \bX_n] - \E[\exp(it_1W)]|]+|\E[\exp\{i(t_2S_n(\bX_n))\}]-\E[\exp\{i(t_2S)\}]|,
  \end{align*}
  where the last inequality follows from the fact that all
  characteristic functions satisfy $|\varphi|\leq 1$.  This expression
  converges to zero in probability due to the convergence of $S_n(\bX_n)$ and
  $W_n(\bZ_n) \mid \bX_n$ respectively. Therefore, we have:
	\[(W_n(\bZ_n), S_n(\bX_n)) \to F_WF_S\]
\end{proof}

\begin{lemma}\label{lem:mean}
	$\displaystyle \lim_{n \to \infty} \E(\hat{\tau}_k)-\tau_k=O\left(n^{-1}\right)$
\end{lemma}
\begin{proof}
  We bound the bias of $\E(\hat{\tau}_k)$ by appealing to Theorem~1 of
  \cite{imai2019experimental}, which implies,
  \begin{align}
		|\E(\hat{\tau}_k)-\tau_k| \leq &
                                                 \left|K\E\left[\int_{F(c_{k}(s))}^{F(\hat{c}_{k}(s))}
                                                 \E(Y_i(1)-Y_i(0) \mid
                                                 s(\bX_i)=F^{-1}(x))\d
                                                 x\right]\right|\nonumber\\&+
    \left|K\E\left[\int_{F(c_{k-1}(s))}^{F(\hat{c}_{k-1}(s))}
    \E(Y_i(1)-Y_i(0) \mid s(\bX_i)=F^{-1}(x))\d x\right]\right|. \label{eq:boundtauk}
\end{align}
By the definition of $\hat{c}_{k}(s)$, $F(\hat{c}_k(s))$ is the
$nk/K$th order statistic of $n$ independent uniform random variables,
and therefore, follows the Beta distribution with the shape and scale
parameters equal to $nk/K$ and $n-nk/K+1$, respectively.

Now, by Assumption~\ref{asm:continuity}, we can compute the
first-order Taylor expansion of
$\int^x_a \E(Y_i(1)-Y_i(0) \mid s(\bX_i)=F^{-1}(x)) \d x$:
\begin{align*}
|\E(\hat{\tau}_k)-\tau_k| \leq &
\left|K\E\left[a_0\{F(\hat{c}_{k}(s))-F(c_{k}(s))\}+o(F(\hat{c}_{k}(s))-F(c_{k}(s)))\right]\right|\\&+
\left|K\E\left[a_1\{F(\hat{c}_{k-1}(s))-F(c_{k-1}(s))\}+o(F(\hat{c}_{k-1}(s))-F(c_{k-1}(s)))\right]\right| \\
                               = & |Ka_0|\left|\frac{nk}{K(n+1)}-\frac{k}{K}\right| +
|Ka_1|\left|\frac{n(k-1)}{K(n+1)}-\frac{k-1}{K}\right|+o\left(n^{-1}\right)\\
  = & O\left(n^{-1}\right).
\end{align*}
\end{proof}

Now, using these two lemmas, we prove the main result.  Letting
$u(s)=\E[U_i \mid s(\bX_i)=s]$, we decompose $\hat\tau_k$ into two
parts,
\begin{align}
	\hat{\tau}_k  
  \ = \ &\underbrace{\frac{1}{n} \sum_{\frac{(k-1)n}{K} < i \leq \frac{kn}{K}} U_{[i,n]} - u(s(\bX_{[i,n]}))}_{\hat{\tau}^{(1)}_k} + \underbrace{\frac{1}{n} \sum_{i=1}^n \mathbf{1}\left\{\frac{(k-1)n}{K} < i \leq \frac{kn}{K}\right\} u(s(\bX_{[i,n]}))}_{\hat{\tau}^{(2)}_k}.\label{eq:term_separate}
\end{align}

Consider the first term.  By the general theory of induced order
statistics presented in \cite{bhattacharya1974convergence},
$U_{[i,n]} - u(s(\bX_{[i,n]}))$ for $i=1,\cdots,n$ are independent of
one another conditional on $\bX_n= (\bX_{[1,n]},\cdots,
\bX_{[n,n]})$. Define the random variables $Z_{[i,n]}$ as distributed
according to the joint conditional distribution
$U_{[i,n]} - u(s(\bX_{[i,n]})) \mid \bX_n$. Then, we have
\[\hat{\tau}^{(1)}_k \ = \  \frac{1}{n}
  \sum_{\frac{(k-1)n}{K} < i \leq \frac{kn}{K}} Z_{[i,n]},\] where
$Z_{[i,n]}$ are conditionally independent and $\E[Z_{[i,n]}]=0$ by
construction. Therefore, by Assumption~\ref{asm:moments}, we can
utilize the Berry-Esseen Theorem. Define:
\begin{align*}
	\sigma_1^2(n) &=  \frac{1}{n}\sum_{\frac{(k-1)n}{K}<i \leq \frac{kn}{K}} \V(Z_{[i,n]})\\
	\rho_1(n) &=  \frac{1}{n}\sum_{\frac{(k-1)n}{K}<i \leq \frac{kn}{K}} \E(|Z_{[i,n]}|^3)
\end{align*}
 Then the Berry-Esseen Theorem states that for $W \sim N(0,1)$, we have:
\begin{equation*}
	d\left(\frac{\sqrt{n}\hat{\tau}^{(1)}_k}{\sqrt{\sigma_1^2(n)}}, W\right) \leq \frac{C_0}{\sqrt{n}}\left(\sigma_1^2(n)\right)^{-3/2}\rho_1(n)
\end{equation*}
where $d(\cdot,\cdot)$ is the Kolmogorov distance. Now define the asymptotic variance and third moment by:
\begin{align*}
	\sigma_1^2 &= \lim_{n \to \infty} \sigma_1^2(n)\\
	\rho_1 &= \lim_{n \to \infty} \rho_1(n)
\end{align*}
Both quantities exist by the strong law of large numbers for functions
of order statistics (see Theorem~4 of
\citet{wellner1977glivenko}). Specifically, by the strong law,
$\sigma_1^2$ and $\rho_1$ does not depend on $\bX_n$ for all but at most a measure
zero set of $\bX_n$. Therefore, the Berry-Esseen theorem implies that:
\begin{equation}
  \sqrt{n}\hat{\tau}_k^{(1)} \mid \bX_n \xrightarrow{d} N(0,\sigma_1^2)\;\;\text{with probability 1}\label{eq:part_1}
\end{equation}

Next, consider the second term of Equation~\eqref{eq:term_separate}.
To prove the convergence of this summation of a function of order
statistics, we utilize Theorem~1 and Example~1 from
\cite{shorack1972functions}, which we restate in our notation below:
\begin{theorem}[\cite{shorack1972functions}]\label{thm:ordstat_clt}
  Consider an independently and identically distributed random sample
  $X_1,\cdots, X_n$ of size $n$ from a cumulative distribution
  function $F$, and a function of bounded variation $g$ such that
  $\E[g(X)^3]<\infty$. Define:
  \[T_n = \frac{1}{n} \sum_{i=1}^n J\l(\frac{i}{n}\r)g(X_{[i,n]})\]
  where $X_{[i,n]}$ is the $i$th order statistics of the sample, and $J$
  is a function that is continuous except at a finite number of points
  at which $g(F^{-1})$ is continuous. Suppose that there exists
  $\delta>0$ such that: 
  \[|J(t)| \leq M(t(1-t))^{-\frac{1}{6} +\delta} \;\; \forall 0<t<1\]
  Then, we have:
  \[\sqrt{n}(T_n -\E[T_n]) \xrightarrow{d} N(0,\sigma^2)\]
  where $\sigma^2=\lim_{n \to \infty} n\V(T_n)<\infty$.
\end{theorem}

Now, set $X_i = s(\bX_i)$, $g(\cdot)=u(\cdot)$, and
$J(t)=\mathbf{1}\left\{(k-1)n/K < tn \leq kn/K\right\}$. Then, we
have $T_n=\hat{\tau}^2_k$. Assumption~\ref{asm:moments} guarantees
$\E[g(X)^3]<\infty$. The function $J(t)$ is discontinuous only at the
quantile points $t=k/K$ and $t=\frac{k-1}{K}$, and Assumption
\ref{asm:continuity} guarantees the continuity of $g(F^{-1})$ at those
points. The function $J$ clearly satisfies the bounding condition with
$\delta=1/6$ and $M=1$. Therefore, define the asymptotic variance as
$\sigma_2^2 = \lim_{n \to \infty} n\V(\hat{\tau}^{(2)}_k),$ and we can
utilize Theorem~\ref{thm:ordstat_clt} to show the following
convergence:
\begin{equation}
  \sqrt{n}(\hat{\tau}^{(2)}_k -\E(\hat{\tau}^{(2)}_k)) \xrightarrow{d} N(0,\sigma_2^2) \label{eq:part_2}
\end{equation}

Now, we aim to combine the results given in
Equations~\eqref{eq:part_1}~and~\eqref{eq:part_2}.  Using
Lemma~\ref{lem:mean}, we can replace $\tau_k$ with $\E(\hat{\tau}_k)$
by adding a small bias term. Then, we have
\begin{align}
\sqrt{n}(\hat{\tau}_k -\tau_k)
  &= \sqrt{n}(\hat{\tau}^{(1)}_k -\E(\hat{\tau}^{(1)}_k))+
    \sqrt{n}(\hat{\tau}^{(2)}_k -\E(\hat{\tau}^{(2)}_k)) +
    O\left(n^{-1/2}\right)\nonumber\\
  & \xrightarrow{d} N(0, \sigma_1^2+\sigma_2^2)\label{eq:GATEsamp_alt}
\end{align}
where the last line follows from the application of
Lemma~\ref{lem:conditional_converge} to the convergence results given in
Equations~\eqref{eq:part_1}~and~\eqref{eq:part_2}.

Equivalently, we can write Equation~\eqref{eq:GATEsamp_alt} as,
\[\sqrt{n}\frac{\hat{\tau}_k -\tau_k}{\sqrt{\sigma_1^2+\sigma_2^2}} \stackrel{d}{\to} N(0,1)\]
Now, note that by the law of total variance, we have that
\begin{align}
  n\V(\hat{\tau}_k)  &= n\E[\V(\hat{\tau}^{(1)}_k+\hat{\tau}^{(2)}_k \mid \bX_n)] + n\V[\E(\hat{\tau}^{(1)}_k+\hat{\tau}^{(2)}_k \mid \bX_n)]\nonumber\\&=n\E[\V(\hat{\tau}^{(1)}_k \mid \bX_n)] + n\V(\hat{\tau}^{(2)}_k)\nonumber\\&\to \sigma_1^2 +\sigma_2^2\label{eq:var_exp}
\end{align}
Therefore, by Slutsky's lemma, we have that:
\[\frac{\hat{\tau}_k -\tau_k}{\sqrt{\V(\hat{\tau}_k)}} \stackrel{d}{\to} N(0,1)\] \qed

\section{Proof of Proposition \ref{prop:continuity_tight}}
\label{app:continuity_tight}

We prove this proposition by finding an example that satisfies it.
Define $t(x)=\E(Y_i(1)-Y_i(0) \mid s(\bX_i)=F^{-1}(x))$. Then,
consider a scoring function $s$ and a population such that:
\begin{equation*}
	t(x) =
	\begin{cases}
		2 & x\geq  F(c_{k}(s))\\
		1 & x< F(c_{k}(s))
	\end{cases}
\end{equation*}
Note that $t(x)$ is bounded everywhere but has a discontinuity. By
definition of $\hat{c}_{k}(s)$, $F(\hat{c}_{k}(s))$ follows the Beta
distribution with the shape and scale parameters equal to $nk/K$ and
$n-nk/K+1$, respectively. Therefore, we have the following normal
approximation:
\begin{equation*}
\sqrt{n+1}\left(F(\hat{c}_{k}(s))-\frac{nk}{K(n+1)}\right)\xrightarrow{d}
N\left(0,\frac{k}{K}\left(1-\frac{k}{K}\right)\right)
\end{equation*}
In particular, as $n\to \infty$, $F(\hat{c}_k(s))$ is distributed
approximately symmetric around $F(c_k(s))=\frac{k}{K}$ with an error
of $O(n^{-1})$ and has a standard deviation of $O(n^{-1/2})$.  Thus, we
have,
\begin{align*}
	\E(\hat{\tau}_k)-\tau_k&= K\E\left[\int_{F(c_{k}(s))}^{F(\hat{c}_{k}(s))} f(x)\d x\right]+ K\E\left[\int_{F(c_{k-1}(s))}^{F(\hat{c}_{k-1}(s))} f(x)\d x\right]\\&= (2-1)O\l(n^{-1/2}\r)+(1-1)O\l(n^{-1/2}\r)+O\l(n^{-1}\r)\\&=O\l(n^{-1/2}\r)
\end{align*}
We can now conclude
$\sqrt{n}(\E(\hat{\tau}_k)-\tau_k) \not \rightarrow 0$. \qed

\section{Proof of Theorem \ref{thm:heterotest}}
\label{app:heterotest}
We wish to prove that for
$\hat{\bm{\tau}}=(\hat{\tau}_1,\cdots,\hat{\tau}_K)$,
$\bm{\tau}=(\tau_1,\cdots,\tau_K)$, and
$\bm{\Sigma}_n = \V(\hat{\bm{\tau}})$, we have:
\[\bm{\Sigma}_n^{-1/2} (\hat{\bm{\tau}} -\bm{\tau}) \stackrel{d}{\to} N(0,\bm{I})\]
where $\bm{I}$ is the $K\times K$ identity matrix.

By Equation~\eqref{eq:GATEsamp_alt} in the proof of
Theorem~\ref{thm:GATEsamp}, for all $k=1,\cdots,k$ we have,
\[\sqrt{n}(\hat{\tau}_k -\tau_k) \stackrel{d}{\to} N(0,\sigma_k^2),\]
where $\sigma_k^2=\lim_{n\to \infty} n\V(\hat{\tau}_k)$. To prove the
multi-dimensional result, we utilize the Cramer-Wold device, which we
restate below:
\begin{theorem}[\cite{cramer1936some}]\label{thm:cramer_wold}
  Let $\bX_n=(X_{n1},\cdots, X_{nk})$ and $\bX=(X_{1},\cdots, X_{k})$ be $k$-dimensional random vectors. Then $\bX_n \to \bX$ if and only if for all $(t_1,\cdots, t_k) \in \mc{R}^k$, we have:
  \[\sum_{i=1}^k t_i X_{ni} \xrightarrow{d} \sum_{i=1}^k t_iX_i\]
\end{theorem}

Now, consider $\bt=(t_1,\cdots, t_K) \in \mc{R}^K$ and
$\hat{\tau}_{\bt} = \sum_{k=1}^K t_k \hat{\tau}_{k}$. Then, we can write
$\hat{\tau}_{\bt}$ as:
\begin{equation*}
	\hat{\tau}_{\bt} = \frac{1}{n}\sum_{i=1}^n \left(\sum_{k=1}^K t_k\mathbf{1} \left\{\frac{(k-1)n}{K} < i \leq \frac{kn}{K}\right\}\right) U_{[i,n]}
\end{equation*}
where $U_{[i,n]}$ is defined in Equation~\eqref{eq:unitdef}.  We use
the same proof strategy as the one used to prove
Theorem~\ref{thm:GATEest}. We define $u(s)=\E[U_i \mid s(\bX_i)=s]$
and write:
\begin{align*}
	\hat{\tau}_{\bt}  
  \ = \ &\underbrace{\frac{1}{n} \sum_{k=1}^K t_k\sum_{\frac{(k-1)n}{K} < i \leq \frac{kn}{K}} U_{[i,n]} - u(s(\bX_{[i,n]}))}_{\hat{\tau}^{(1)}_{\bt}}\nonumber\\& + \underbrace{\frac{1}{n} \sum_{i=1}^n \left(\sum_{k=1}^K t_k\mathbf{1}\left \{\frac{(k-1)n}{K} < i \leq \frac{kn}{K}\right\}\right) u(s(\bX_{[i,n]}))}_{\hat{\tau}^{(2)}_{\bt}}
\end{align*}
Using Lindberg's Central Limit Theorem for the conditional
distribution of $\hat{\tau}_{\bt}^{(1)}$ given $\bX_n $ and applying
Theorem~\ref{thm:ordstat_clt} to $\hat{\tau}_{\bt}^{(2)}$ yield,
\begin{align*}
	\sqrt{n}(\hat{\tau}^{(1)}_{\bt} - \E(\hat{\tau}^{(1)}_{\bt})) \mid \bX_n
  & \ \stackrel{d}{\to} \ N(0,\sigma_1^2)\\
	\sqrt{n}(\hat{\tau}^{(2)}_{\bt} - \E(\hat{\tau}^{(2)}_{\bt})) & \
                                                                    \stackrel{d}{\to}
                                                                    \ N(0,\sigma_2^2)
\end{align*}
where
$\sigma_1^2 = \displaystyle \lim_{n\to\infty} n
\V(\hat{\tau}^{(1)}_{\bt} \mid \bX_n)$ and
$\sigma_2^2 = \displaystyle \lim_{n\to \infty}
n\V(\hat{\tau}_{\bt}^{(2)})$. This implies,
\begin{align*}
	&\sqrt{n}(\hat{\tau}_{\bt} -\tau_{\bt}) \ = \
   \sqrt{n}\left(\hat{\tau}^{(1)}_{\bt}
   -\E[\hat{\tau}^{(1)}_{\bt}]\right)+
   \sqrt{n}\left(\hat{\tau}^{(2)}_{\bt} -\E[\hat{\tau}^{(2)}_{\bt}]\right)
   + O\left(n^{-1/2}\right) \xrightarrow{d} N(0,
   \sigma_1^2+\sigma_2^2)
\end{align*}
where the result follows from Lemma~\ref{lem:conditional_converge}.
Since we can easily show
$n\V(\hat{\tau}_t) \to \sigma_1^2 +\sigma_2^2$, we have:
$\sqrt{n}(\hat{\tau}_{\bt} -\tau_{\bt}) \to N(0, \lim_{n\to \infty}
n\V(\hat{\tau}_{\bt}))$.  Therefore, by the Cramer-Wold device
(Theorem~\ref{thm:cramer_wold}), we have
$\sqrt{n}(\hat{\bm{\tau}}-\bm{\tau})\to N(0, \lim_{n\to \infty} n\bm{\Sigma}_n).$
Finally, Slutsky's Lemma implies the desired result,
\[\bm{\Sigma}_n^{-1/2}(\hat{\bm{\tau}}-\bm{\tau}) \to N(0, \bm{I}).\]

To derive the expression for the covariance matrix
$\bm{\Sigma}_{kk'}$, we utilize the same approach as the one used in
the proof of Theorem~\ref{thm:GATEest}. We first apply the law of
total covariance to obtain: 
\begin{align}
	 & \Cov(\hat\tau_k, \hat\tau_{k'}) \nonumber \\
	= & \ K^2\Cov\l(\frac{1}{n} \sum_{i=1}^n \{Y^*_{ki}(1) -
	Y^*_{ki}(0)\},\frac{1}{n} \sum_{i=1}^n \{Y^*_{k'i}(1) -
	Y^*_{k'i}(0)\}\r) \nonumber \\
	& +K^2\E\left[\Cov\left\{\frac{1}{n_1}\sum_{i=1}^n \left(\hat{f}_k(\bX_i)-\frac{1}{K}\right)T_iY_i(1)-\frac{1}{n_0}\sum_{i=1}^n\left(\hat{f}_k(\bX_i)-\frac{1}{K}\right)(1-T_i)Y_i(0), \right.\right.\nonumber\\ & \hspace{.25in}\left. \left.\frac{1}{n_1}\sum_{i=1}^n \left(\hat{f}_{k'}(\bX_i)-\frac{1}{K}\right)T_iY_i(1)-\frac{1}{n_0}\sum_{i=1}^n\left(\hat{f}_{k'}(\bX_i)-\frac{1}{K}\right)(1-T_i)Y_i(0)\biggm\vert \bX, \{Y_i(1), Y_i(0)\}_{i=1}^n\right\}\right]. \label{eq:covar_decomp}
\end{align}
Applying Neyman's finite sample variance analysis to the second term
shows that this term is equal to:
\begin{align}
	K^2\E\left\{\frac{1}{n} \l( \frac{n_0}{n_1} S_{kk'1}^{*2} + \frac{n_1}{n_0} S_{kk'0}^{*2}
	+ 2 S^*_{kk'01}\r)\r\},\label{eq:covar2_decomp}
\end{align}
where
$S^*_{kk'01} = \sum_{i=1}^n (Y^*_{ki}(0) - \overline{Y^*_{k}(0)})(Y^*_{k'i}(1)
- \overline{Y^*_{k'}(1)})/(n-1)$.  Since $Y^*_{ki}(t)$ and $Y^*_{k'j}(t)$
are correlated, we have:
\begin{align}
	&\Cov\l(\frac{1}{n} \sum_{i=1}^n (Y^*_{ki}(1) - Y^*_{ki}(0)),\frac{1}{n} \sum_{i=1}^n (Y^*_{k'i}(1) - Y^*_{k'i}(0))\r)
	\ \nonumber \\= &\ {\rm Cov}(Y^*_{ki}(1) -Y^*_{ki}(0), Y^*_{k'j}(1) -
	Y^*_{k'j}(0)) + \frac{1}{n}   \E(S_{kk'1}^{*2}  + S_{kk'0}^{*2} - 2 S^*_{kk'01}). \label{eq:covar1_decomp}
\end{align}
For $k\neq k'$, we have: 
\begin{align*}
	&{\rm Cov}(Y^*_{ki}(1) -Y_{ki}(0), Y^*_{k'j}(1) -
	Y_{k'j}(0))\\
	= \ &  {\rm Cov}\left(\left(\hat{f}_k(\bX_i)-\frac{1}{K}\right)\tau_i, \left(\hat{f}_{k'}(\bX_j)-\frac{1}{K} \right)\tau_j\right) \nonumber\\
  = \ & \left
        (1-\frac{2}{K}\right)\Cov(\hat{f}_k(\bX_i)\tau_i,\hat{f}_{k'}(\bX_j)\tau_j)
        -
        \frac{1}{K}\Cov((1-\hat{f}_k(\bX_i))\tau_i,\hat{f}_{k'}(\bX_j)\tau_j) \\
            &
              -\frac{1}{K}\Cov((1-\hat{f}_{k'}(\bX_i))\tau_i,\hat{f}_k(\bX_j)\tau_j)
   \\
   = \ &  \left(1-\frac{2}{K}\right)\left(\frac{1}{K^2}\kappa_{kk'11}
         - \frac{1}{K^2}\kappa_{k1}\kappa_{k'1}\right)\\ & - 
                                                           \frac{1}{K}\left\{\frac{n/K(n-n/K-1)}{n(n-1)}\kappa_{kk'01}+\frac{n/K(n-n/K-1)}{n(n-1)}\kappa_{kk'10}-\frac{1}{K^2}\kappa_{k1}\kappa_{k'0}-\frac{1}{K^2}\kappa_{k0}\kappa_{k'1}\right\}
   \\
   = \ & \frac{1}{K^3}\left\{(K-2)\left(\kappa_{kk'11}-\kappa_{k1}\kappa_{k'1}\right)-\frac{Kn-n-1}{n-1}\left(\kappa_{kk'10}+\kappa_{kk'01}\right)+\kappa_{k1}\kappa_{k'0}+\kappa_{k0}\kappa_{k'1}\right\}.        
\end{align*}
For $k=k'$, we have
\begin{align*}
	&{\rm Cov}(Y_{ki}^*(1) -Y_{ki}^*(0), Y_{kj}^*(1) -
	Y_{kj}^*(0))\\
	= \ &  {\rm Cov}\left(\left(\hat{f}_k(\bX_i)-\frac{1}{K}\right)\tau_i, \left(\hat{f}_{k}(\bX_j)-\frac{1}{K} \right)\tau_j\right) \nonumber\\
	= \ &  \left(1-\frac{2}{K}\right)\Cov(\hat{f}_k(\bX_i)\tau_i,\hat{f}_{k}(\bX_j)\tau_j)  -\frac{2}{K}\Cov((1-\hat{f}_k(\bX_i))\tau_i,\hat{f}_{k}(\bX_j)\tau_j)\\
 = \ &
       \left(1-\frac{2}{K}\right)\left\{\frac{n-K}{K^2(n-1)}\kappa_{k11}
       - \frac{1}{K^2}\kappa_{k1}^2\right\}
       -\frac{2}{K}\left\{\frac{n(K-1)}{K^2(n-1)}\kappa_{k01}-\frac{1}{K^2}\kappa_{k1}\kappa_{k0}\right\}
  \\
    = \ &
          \frac{1}{K^3}\left\{(K-2)\left(\frac{n-K}{n-1}\kappa_{k11}-\kappa_{k1}^2\right)-\frac{2n(K-1)}{(n-1)}\kappa_{k01}+2\kappa_{k1}\kappa_{k0}\right\}.       
\end{align*}
Substituting
Equations~\eqref{eq:covar2_decomp}~and~\eqref{eq:covar1_decomp} into
Equation \ref{eq:covar_decomp}, we obtain the desired covariance
expression.

\qed

\section{Proof of Theorem \ref{thm:consistest}}
\label{app:consistest}

The proof of Theorem~\ref{thm:heterotest} above establishes that
$\bm{\Sigma}^{-1/2}\bm{\hat\tau}$ is asymptotically normally
distributed with the identity variance matrix $\bm{I}$.  For
simplicity, throughout this proof, we will assume that
$\bm{\Sigma}^{-1/2}\bm{\hat\tau}$ is exactly normally distributed with
unknown mean $\bm{\theta}=(\tau_{1}-\tau,\cdots,\tau_{K}-\tau)$, i.e.,
$\bm{\Sigma}^{-1/2}\bm{\hat\tau} \sim N(\bm{\theta},\bm{I})$.

Let the likelihood of the data $\bm{\hat\tau}$ under the null and
alternative hypotheses as $L_{\bm{\hat\tau}}(H_0^C)$ and
$L_{\bm{\hat\tau}}(H_1^C)$.  Under the asymptotic normal assumption,
the likelihood ratio is given by:
\begin{equation*}
\frac{L_{\bm{\hat\tau}}(H_0^C)}{L_{\bm{\hat\tau}}(H_1^C)} = \begin{cases}
\exp\left\{(\bm{\hat\tau}-\bm{\mu}_1(\bm{\hat\tau}))^\top\bm{\Sigma}^{-1}(\bm{\hat\tau}-\bm{\mu}_1(\bm{\hat\tau}))\right\}
& \bm{\theta} \in \Theta_0\\
\exp\left\{-(\bm{\hat\tau}-\bm{\mu}_0(\bm{\hat\tau}))^\top\bm{\Sigma}^{-1}(\bm{\hat\tau}-\bm{\mu}_0(\bm{\hat\tau}))\right\}& \bm{\theta}  \in \Theta_1
\end{cases}
\end{equation*}
Where $\bm{\mu}_i(\bm{\hat\tau})$ are the optimal mean vectors given
data $\bm{\hat\tau}$ for region $j$ of the hypothesis test, and is the
solution to the following optimization problems for $j \in \{0,1\}$:
\begin{equation*}
\bm{\mu}_j(\bm{\hat\tau}) = \argmin_{\bm{\mu} \in \Theta_j} \|\bm{\hat\tau}-\bm{\mu}\|^2
\end{equation*}
We can identify the optimal means ($\bm{\mu}_1,\bm{\mu}_0$) for each
region of the hypothesis test through this optimization problem
because the multivariate normal distribution is spherical and
symmetric.

We use
$(\bm{\hat\tau}-\bm{\mu}_0(\bm{\hat\tau}))^\top
\bm{\Sigma}^{-1}(\bm{\hat\tau}-\bm{\mu}_0(\bm{\hat\tau}))$ as our test
statistic. Note that when $\bm{\hat\tau} \in \Theta_0$, the statistic
is always 0, so the null hypothesis is never rejected and thus we are
consistent.  Given that we have a composite test, we are interested in
finding the uniformly most powerful test.  This requires calculating
the size of a test $\alpha$, as a function of the critical value
$C(\alpha)$:
\begin{align*}
\alpha &=\sup_{\bm{\theta} \in \Theta_0} \Pr(
         (\bm{\hat\tau}-\bm{\mu}_0(\bm{\hat\tau}))^\top \bm{\Sigma}^{-1}(\bm{\hat\tau}-\bm{\mu}_0(\bm{\hat\tau}))>C(\alpha)
         \mid \bm{\theta} )
\end{align*}
Since the supremum must occur at the boundary $\partial \bm{\Theta}_0$
of the polytope $\bm{\Theta}_0$ the set $\Theta_0$, the probability of
exceeding $C(\alpha)$ is maximized when the solid angle of the
$\bm{\Theta}_0$ region is minimized.  By considering the shape of the
polytope $\bm{\Theta}_0$, we recognize that the boundary points, which
minimize the solid angle, are precisely those on the boundary when all
constraints are active:
\begin{align*}
\alpha &=\sup_{t} \Pr(
         (\bm{\hat\tau}-\bm{\mu}_0(\bm{\hat\tau}))^\top\bm{\Sigma}^{-1}(\bm{\hat\tau}-\bm{\mu}_0(\bm{\hat\tau}))>C(\alpha)
         \mid \tau_{1}-\tau=\cdots =\tau_{K}-\tau=t ).
\end{align*}
We now note that we have translational invariance on this boundary,
i.e., all points along $\tau_{1}-\tau=\cdots =\tau_{K}-\tau$ have the
same probability, yielding,
\begin{align*}
\alpha &= \Pr( (\bm{\hat\tau}-\bm{\mu}_0(\bm{\hat\tau}))^\top \bm{\Sigma}^{-1}(\bm{\hat\tau}-\bm{\mu}_0(\bm{\hat\tau}))>C(\alpha) \mid\tau_{1}-\tau=\cdots =\tau_{K}-\tau=0 )
\end{align*}
Therefore, to identify the value of $\alpha$, we just need the CDF of
the statistic
$(\bm{\hat\tau}-\bm{\mu}_0(\bm{\hat\tau}))^\top
\bm{\Sigma}^{-1}(\bm{\hat\tau}-\bm{\mu}_0(\bm{\hat\tau}))$ when
$\hat\tau \sim N(\bm{0}, \bm{\Sigma})$. This can be easily estimated
using Monte Carlo simulation. \qed

\bigskip

\section{Proof of Theorem~\ref{thm:GATEestcv}}
\label{app:GATEestcv}

The derivation of bias is essentially identical to that given in
Supplementary Appendix~\ref{app:GATEest} and thus is omitted. To
derive the variance, we first introduce the following useful lemma,
adapted from \cite{nadeau2000inference}.
\begin{lemma}
	\label{lemma:cvcov}
	\begin{eqnarray*}
		\E(S_{Fk}^2) & = & \V(\hat{\tau}_k^\ell)-\Cov(\hat{\tau}_k^\ell,\hat{\tau}_k^{\ell'}),\\
		\V(\hat \tau_{k}(F,n-m)) & = & \frac{\V(\hat{\tau}_k^\ell)}{L}+ \frac{L-1}{L}\Cov(\hat{\tau}_k^\ell,\hat{\tau}_k^{\ell'}).
	\end{eqnarray*}
	where $\ell \ne \ell'$.
\end{lemma}
The lemma implies,
\begin{equation*}
	\V(\hat \tau_{k}(F,n-m)) \ = \ \V(\hat{\tau}_k^\ell)-\frac{L-1}{L}\E(S_{Fk}^2) .
\end{equation*}
We then follow the same process of derivation as in
Appendix~\ref{app:GATEest} for the first term. The only difference
occurs in the derivation of the covariance term:
\begin{align*}
	&  {\rm Cov}(Y_{ki}^\ell(1) -Y_{ki}^\ell(0), Y_{kj}^\ell(1) -
	Y_{kj}^\ell(0)) \nonumber\\
	=& \E_\ell\l[ {\rm Cov}_{\bX,Y}\l(\hat{f}_k^\ell(\bX_i^{(\ell)})\tau_i, \hat{f}_k^\ell(\bX_j^{(\ell)})\tau_j\r)\r] + \Cov_\ell\l[\E_{\bX,Y}[\hat{f}_k^\ell(\bX_i^{(\ell)})\tau_i],\E_{\bX,Y}[\hat{f}_k^\ell(\bX_j^{(\ell)})\tau_j]\r] \nonumber\\ = & \E_\ell\l[\frac{(n-K)\kappa^\ell_{k11}}{K^2(n-1)} -	\frac{1}{K^2}(\kappa^\ell_{k1})^2\r]+\frac{1}{K^2}\V_\ell\l(\kappa_{k1}^\ell\r).
\end{align*}
\qed

 \bigskip

%

\section{Proof of Consistency of $\widehat{\E(S^2_{Fk})}$}
\label{app:consistent}
We show that $\widehat{\E(S^2_{Fk})}$ is consistent as $L$ approaches
infinity under the assumption that the fourth moments
$\E(Y_i(t)^4)<\infty$ for $t=0,1$ and a sufficiently large value of
$m$. Theorem~\ref{thm:GATEestcv} implies,
\begin{align*}
	\V(\hat \tau_{k}(F,n-m)) &= \ K^2\left(\frac{\E(S_{Fk1}^2)}{m_1} +
	\frac{\E(S_{Fk0}^2)}{m_0}\right) \\&+\E_\ell\l[\frac{(n-K)\kappa^\ell_{k11}}{K^2(n-1)} -	\frac{1}{K^2}(\kappa^\ell_{k1})^2\r]+\V\l(\kappa_{k1}^\ell\r)-\frac{L-1}{L}\E(S_{Fk}^2),
\end{align*}
Now, define:
\[\widehat{\V(\hat \tau_{k}(F,n-m))}=K^2\left(\frac{\widehat{\E(S_{Fk1}^2)}}{m_1} +
\frac{\widehat{\E(S_{Fk0}^2)}}{m_0}\right)+\frac{(n-K)\widehat{\E[\kappa^\ell_{k11}]}}{K^2(n-1)} -	\frac{1}{K^2}(\widehat{\E[\kappa^\ell_{k1}]})^2+\widehat{\V\l(\kappa_{k1}^\ell\r)}\]
By construction, we have that as $m \to \infty$:
\[\frac{\widehat{\V(\hat \tau_{k}(F,n-m))}}{\V(\hat \tau_{k}(F,n-m)) +  \frac{L-1}{L}\E(S_{Fk}^2)} \xrightarrow{p} 1\]
Applying Lemma 1 from \cite{nadeau2000inference} to $\hat \tau_{k}(F,n-m)$ gives
\begin{equation*}
	\V(\hat \tau_{k}(F,n-m)) \geq \ \E(S^2_{Fk}).
\end{equation*}
Therefore, we have:
\[\lim_{m \to \infty} \frac{\widehat{\V(\hat \tau_{k}(F,n-m))}}{\E(S^2_{Fk}) +  \frac{L-1}{L}\E(S_{Fk}^2)} \geq 1\]
Note that we can write $\widehat{\E(S^2_{Fk})}$ as:
\[\widehat{\E(S^2_{Fk})} \ = \ \min\l(S^2_{Fk}, \widehat{\V(\hat \tau_{k}(F,n-m))}\r).\]
By definition of $S^2_{Fk}$, if the fourth moments of $Y_i(t)$ exist,
we have $\V(S^2_{Fk})=O(L^{-1})$, and thus as $L \to \infty$:
\begin{equation*}
	\frac{S^2_{Fk}}{\E(S^2_{Fk})} \xrightarrow{p} 1 
\end{equation*}
Let $\epsilon>0$. There exists $L_0$ such that for all $L>L_0$,
$|\frac{S^2_{Fk}}{\E(S^2_{Fk})}-1|<\epsilon$. Similarly, there exists
$m_0$ such that for all  $\widehat{\V(\hat \tau_{k}(F,n-m))} > (1-\epsilon)\E(S^2_{Fk}) \;\forall m >m_0$. Therefore, for all $m>m_0$, we have that:
\begin{equation*}
	\lim_{L \to \infty}
        \frac{\widehat{\E(S^2_{Fk})}}{\E(S^2_{Fk})} \ = \ \lim_{L \to
          \infty}\frac{\min\l(S^2_{Fk}, \widehat{\V(\hat
            \tau_{k}(F,n-m))}\r)}{\E(S^2_{Fk})} = \lim_{L \to \infty}\frac{S^2_{Fk}}{\E(S^2_{Fk})} \xrightarrow{p} 1
\end{equation*}
\qed

\section{Proof of Theorem \ref{thm:GATEsampcv}}
\label{app:GATEsampcv}

We first shows the bias is small.  
\begin{lemma}\label{lem:mean_cv}
	$\displaystyle \lim_{n \to \infty} \left|\hat{\tau}_k(F,n-m)-\tau_k(F,n-m)\right|=O\left(m^{-1}\right)$
\end{lemma}
\begin{proof}
\begin{align*}
	\left|\hat{\tau}_k(F,n-m)-\tau_k(F,n-m)\right| & \ \leq \
	\frac{1}{L}\sum_{\ell=1}^L
	|\E(\hat\tau_k^\ell(F,n-m))-\tau_k^\ell(F,n-m)|\\
	& \ = \  \frac{1}{L} \sum_{\ell=1}^L
	\E_{\cZ^{-\ell}}\left[O\left(m^{-1}\right)\right] \\
	&  \ = \ O\left(m^{-1}\right).
\end{align*}
The first equality follows because the estimator for each fold
$\hat\tau_k^\ell(F,n-m)$ is equivalent to the non-cross-fitting
estimator under $m$ samples and so Lemma~\ref{lem:mean} is
applicable. The second equality follows from
Assumption~\ref{asm:continuity_cv}.
\end{proof}

 We first write:
\begin{equation*}
	\hat{\tau}_k(F,n-m) = \frac{1}{m} \sum_{i=1}^m \mathbf{1}\left\{\frac{(k-1)m}{K}< i \leq \frac{km}{K}\right\}U_{[i, m]}
\end{equation*}
where $U_{[i, m]} \in \mc{R}$ is defined as,
\begin{align*}
	U_{[i, m]}& :=\frac{1}{L}\sum_{\ell=1}^L
	K\hat{f}_k^{\ell}(\bX_{[i, m]}^{(\ell)})Y_{[i, m]}^{(\ell)}\left(\frac{T_{[i, m]}^{(\ell)}}{q}-\frac{1-T_{[i, m]}^{(\ell)}}{1-q}\right). \label{eq:unitdefcv}
\end{align*}
and $(Y^{(\ell)}_{[i, m]},T^{(\ell)}_{[i, m]},\bX^{(\ell)}_{[i, m]})$ are ordered separately for each split $\ell$ such that:
\[s^{\ell}(\bX^{(\ell)}_{[i, m]})\leq s^{\ell}(\bX^{(\ell)}_{[i, m]})\leq \cdots \leq s^{\ell}(\bX^{(\ell)}_{[i, m]})\]
Now by Assumption \ref{asm:stability}, there exists a fixed scoring rule $s(\bX)$ and corresponding treatment rule $f_k(\bX_i) = \mathbf{1}\{s(\bX_i) > c_{k-1}(s)\}-\mathbf{1}\{s(\bX_i) > c_k(s)\}$ such that we can write:
\begin{align*}
	U_{[i, m]}& =\tilde{U}_{[i, m]} + \epsilon_{[i, m]}\\
	\tilde{U}_{[i, m]} & := \frac{1}{L}\sum_{\ell=1}^L
	Kf_k(\bX_{[i, m]}^{(\ell)})Y_{[i, m]}^{(\ell)}\left(\frac{T_{[i, m]}^{(\ell)}}{q}-\frac{1-T_{[i, m]}^{(\ell)}}{1-q}\right)\\
	\tilde{\tau}_k(F,n-m) &=\frac{1}{m} \sum_{i=1}^m \mathbf{1}\left\{\frac{(k-1)m}{K} <i \leq \frac{km}{K}\right\}\tilde{U}_{[i, m]}
\end{align*}
where 
\begin{align*}
	\E[\epsilon_{[i,m]}]&=\E\left[\frac{1}{L}\sum_{l=1}^L K (f_k(\bX_{[i, m]}^{(\ell)})-\hat{f}_k^{\ell}(\bX_{[i, m]}^{(\ell)})) Y_{[i, m]}^{(\ell)}\left(\frac{T_{[i, m]}^{(\ell)}}{q}-\frac{1-T_{[i, m]}^{(\ell)}}{1-q}\right)\right]\\&\leq\sqrt{\E\left[\left(\frac{1}{L}\sum_{l=1}^L K (f_k(\bX_{[i, m]}^{(\ell)})-\hat{f}_k^{\ell}(\bX_{[i, m]}^{(\ell)})) Y_{[i, m]}^{(\ell)}\left(\frac{T_{[i, m]}^{(\ell)}}{q}-\frac{1-T_{[i, m]}^{(\ell)}}{1-q}\right)\right)^2\right]}\\&= o\left(m^{-1/2}\right)
\end{align*}
Then, we can apply the proof  of Theorem~\ref{thm:GATEsamp} on $\tilde{U}_{[i,m]}$ as $f_k$ is fixed, which gives:
\begin{equation*}
	\frac{\tilde{\tau}_k(F,n-m)-\E[\tilde{\tau}_k(F,n-m)]}{\sqrt{\V(\tilde{\tau}_k(F,n-m))}} \to N(0,1)
\end{equation*}
Since $\V(\tilde{\tau}_k(F,n-m))=\V(\hat{\tau}_k(F,n-m)) + o(m^{-1})$ and $\hat{\tau}_k(F,n-m)=\tilde{\tau}_k(F,n-m) + o(m^{-1/2})$, we have:
\begin{equation*}
	\frac{\hat{\tau}_k(F,n-m)-\tau_k(F,n-m)}{\sqrt{\V(\hat{\tau}_k(F,n-m))}} \to N(0,1)
\end{equation*}
\qed

 \section{Proof of Theorem \ref{thm:heterotestcv}}
 \label{app:heterotestcv}
 The proof follows identically to the proof of Theorem \ref{thm:heterotest} by applying the Cramer-Wold Device in Theorem \ref{thm:cramer_wold} to the sequence $\sum_{k=1}^K t_k \hat{\tau}_k(F,n-m)$.

\end{document}